\documentclass{article}
\usepackage{microtype}
\usepackage{graphicx}
\usepackage{subcaption}
\usepackage{booktabs} 

\usepackage{hyperref}

\usepackage[preprint]{conf2026}

\usepackage{amsmath}
\usepackage{amssymb}
\usepackage{mathtools}
\usepackage{amsthm}

\usepackage[capitalize,noabbrev]{cleveref}

\theoremstyle{plain}
\newtheorem{theorem}{Theorem}[section]
\newtheorem{proposition}[theorem]{Proposition}
\newtheorem{lemma}[theorem]{Lemma}

\theoremstyle{definition}

\newtheorem{assumption}[theorem]{Assumption}
\theoremstyle{remark}

\usepackage[textsize=tiny]{todonotes}

\usepackage[utf8]{inputenc}
\usepackage{booktabs}  
\usepackage{tabularx}             
\usepackage{pifont}   

\usepackage{multirow}
\usepackage[table]{xcolor}

\usepackage{makecell}

\conftitlerunning{Permutation-Invariant Representations for Combinatorial Treatments}

\begin{document}

\twocolumn[
  \conftitle{Orthogonal Uplift Learning with Permutation-Invariant Representations for Combinatorial Treatments
}

\newcommand{\confsetaffiliationnumber}[2]{%
  \ifcsname the@affil#1\endcsname\else
    \newcounter{@affil#1}%
  \fi
  \setcounter{@affil#1}{#2}%
  \ifnum\value{@affiliationcounter}<#2\relax
    \setcounter{@affiliationcounter}{#2}%
  \fi
}
 \confsetsymbol{equal}{*}        
\confsetaffiliationnumber{fi}{1} 
\confsetaffiliationnumber{se}{2} 

  \begin{confauthorlist}
  \confauthor{Xinyan Su}{equal,fi}
    \confauthor{Jiacan Gao}{equal,se}
    \confauthor{Mingyuan Ma}{comp}
    \confauthor{Xiao Xu}{fi}
    \confauthor{Xinrui Wan}{fi}
    \confauthor{Tianqi Gu}{fi}
    \confauthor{Enyun Yu}{fi}
    \confauthor{Jiecheng Guo}{fi}
    \confauthor{Zhiheng Zhang}{shangcai1,shangcai2}
  \end{confauthorlist}

\confaffiliation{fi}{Didi Chuxing, Beijing, China}
\confaffiliation{se}{School of Statistics, East China Normal University, Shanghai, China}

  \confaffiliation{comp}{School of Mathematics and Statistics, Beijing Jiaotong University, Beijing, China}
  \confaffiliation{shangcai1}{School of Statistics and Data Science, Shanghai University
of Finance and Economics, Shanghai 200433, P.R. China}
  \confaffiliation{shangcai2}{Institute of Data Science and Statistics, Shanghai University of Finance and Economics,
Shanghai 200433, P.R. China}

  \confcorrespondingauthor{Zhiheng Zhang}{zhangzhiheng@mail.shufe.edu.cn}

  \confkeywords{Causal Inference, Permutation-Invariant Sets, Treatment Effect Estimation}

  \vskip 0.3in
]

\printAffiliationsAndNotice{\confEqualContribution}

\begin{abstract}
We study uplift estimation for \emph{combinatorial} treatments. Uplift measures the {pure incremental causal effect} of an intervention (e.g., sending a coupon or a marketing message) on user behavior, modeled as a conditional individual treatment effect. Many real-world interventions are \emph{combinatorial}: a treatment is a policy that specifies context-dependent action distributions rather than a single atomic label.
Although recent work considers structured treatments, most methods rely on categorical or opaque encodings, limiting robustness and generalization to rare or newly deployed policies. We propose an uplift estimation framework that aligns treatment representation with causal semantics.
Each policy is represented by the mixture it induces over context--action components and embedded via a {permutation-invariant} aggregation.
This representation is integrated into an {orthogonalized low-rank uplift model}, extending Robinson-style decompositions to learned, vector-valued treatments. We show that the resulting estimator is expressive for policy-induced causal effects, orthogonally robust to nuisance estimation errors, and stable under small policy perturbations.
Experiments on large-scale randomized platform data demonstrate improved uplift accuracy and stability in long-tailed policy regimes.
\end{abstract}

\section{Introduction}\label{intro}

Estimating individualized causal effects plays a central role in data-driven decision-making, enabling practitioners to compare alternative interventions and select those that maximize expected incremental outcomes.
This problem is commonly formalized through the conditional average treatment effect (CATE), which has been extensively studied in statistics, econometrics, and machine learning
\citep{rubin1974estimating, imbens2015causal, athey2016recursive}.
A rich body of work has developed estimation strategies for CATE under binary or low-cardinality treatments, including meta-learners, doubly robust methods, and orthogonalized estimators
\citep{chernozhukov2018double, kunzel2019metalearners,  kennedy2020optimal}.
These approaches have achieved notable success in classical A/B testing and personalization tasks.

\begin{figure*}[t]
    \centering
    \includegraphics[width=1.0\linewidth]{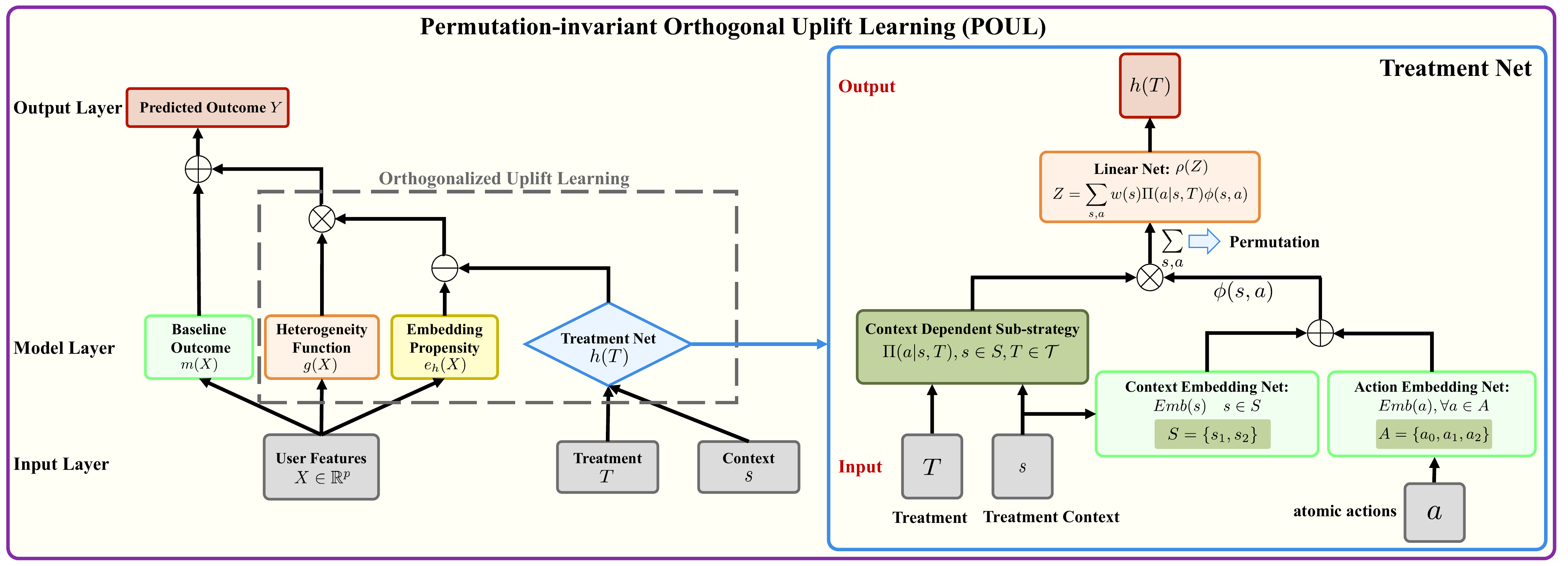}
    \caption{\textbf{Treatment Net: permutation-invariant embedding of combinatorial treatments.}
    A treatment $T$ is observed through its policy specification: for each context $s\in S$ (instantiated here as $\{s_1, s_2\}$), the treatment induces a distribution over atomic actions $a\in \mathcal{A}$ (instantiated as $\{a_0, a_1, a_2\}$).
    The network first learns embeddings for contexts and actions to form atom representations $\phi(s,a)$.
    These atoms are then reweighted by policy probabilities and aggregated via a permutation-invariant sum $z(T)$, ensuring the representation depends only on the induced mixture.
    Finally, a learnable map produces the treatment embedding $h(T)$, which is used by the orthogonalized uplift model (left panel) to estimate the causal effect $\tau(x)$.}
    \label{fig:treatment-net}
\end{figure*}

In many modern applications, however, a treatment is rarely a single atomic action.
Instead, it is a \emph{policy} that specifies how actions are chosen across multiple treatment-side contexts, such as user segments, locations, or stages of a workflow.
For instance, an online platform may deploy a policy that assigns different incentive levels depending on spatiotemporal demand conditions.
Such interventions naturally induce a \emph{structured} or \emph{combinatorial} treatment space, where high-level treatments are composed of reusable lower-level components. A practical difficulty is that these policies are constantly recombined: a new policy is often created by changing only a few rules (e.g., swapping the action for one segment).
This setting differs fundamentally from standard multi-treatment problems, where treatments are modeled as unrelated categorical labels.
Recent work has begun to explore structured treatments using embeddings or graph-based representations
\citep{bica2020estimating, guo2021graphiite}, but these methods typically rely on opaque encodings that do not explicitly reflect the compositional semantics of policies.

A key challenge in combinatorial treatment settings is that the statistical representation of treatments often fails to align with their causal meaning.
In practice, policies are frequently iterated and implemented by different teams, leading to variations in indexing, ordering, or internal identifiers.
Naively encoding each policy as a one-hot label treats such variants as unrelated, even when they induce identical or nearly identical interventions or differ only by minor local modifications.
This mismatch has two critical consequences.
First, it prevents statistical strength from being shared across related policies, exacerbating variance in long-tailed treatment regimes.
Second, it destroys any meaningful notion of proximity between treatments, making it impossible to reason about smooth policy perturbations or to generalize to new strategies.
These issues arise simultaneously in theory and practice, and cannot be resolved by simply increasing model capacity.

Motivated by these observations, we adopt the perspective that a policy affects outcomes primarily through the mixture it induces over context--action components.
Under this view, the identity or ordering of internal identifiers is irrelevant; what matters is how often, and under which contexts, each atomic action is realized in the population.
This perspective suggests that treatment representations should respect natural symmetries of policies, including invariance to re-indexing and ordering, while enabling principled parameter sharing across reusable components.
Formally, each policy induces a probability measure over context--action atoms, and causal effects depend on this induced mixture rather than on arbitrary labels or orderings.
This idea resonates with recent advances in invariant and equivariant learning, where respecting known symmetries is essential for generalization \citep{zaheer2017deepsets, kondor2018generalization}.
However, its implications for causal effect estimation under structured policy treatments have not been fully explored.

Building on this principle, we propose a framework for uplift estimation under combinatorial treatments that combines two key ingredients.
First, we introduce a permutation-invariant treatment embedding that represents a policy by aggregating learned embeddings of its context--action components via a commutative sum.
This construction ensures invariance to re-indexing and ordering, while enabling parameter sharing across a large treatment space.
Second, we embed this representation into an orthogonalized low-rank uplift model, extending the Robinson decomposition to vector-valued treatments.
Orthogonalization removes first-order sensitivity to nuisance components, yielding robustness properties analogous to those in double/debiased machine learning
\citep{robins1994estimation, chernozhukov2018double}.

Our analysis establishes that the proposed embedding family is expressive enough to approximate any continuous causal functional of the policy-induced mixture.
Moreover, the orthogonalized formulation guarantees that errors in baseline outcome regression and treatment propensity estimation enter the CATE estimator only at second order.
Crucially, the representation induces a natural metric on policy space under which both the embedding and the implied uplift are stable.
This provides a principled explanation for empirical generalization to rare or previously unseen policy variants, a phenomenon that categorical treatment representations cannot capture.
From a practical perspective, these properties are essential in large-scale experimentation systems characterized by frequent policy iteration and severe data imbalance. Our main contributions are as follows: 1) We formalize individualized uplift estimation under \emph{combinatorial treatments}, modeling each treatment as a context-dependent policy composed of reusable context--action components.
2) We propose a permutation-invariant treatment embedding that aligns statistical representation with causal semantics and enables information sharing and generalization across large policy spaces.
3) We integrate this representation with an orthogonalized low-rank uplift model and establish second-order robustness to nuisance estimation errors.
4) We validate the framework on large-scale randomized experiments, demonstrating improved uplift accuracy and stability, including strong performance in long-tailed and held-out treatment regimes.

\section{Problem Setup and Method}\label{setup}

We have i.i.d.\ data $\mathcal{D}=\{(X_i,T_i,Y_i)\}_{i=1}^n$,
where $X_i\in\mathbb{R}^p$ are unit covariates,
$T_i\in\mathcal{T}$ is a structured treatment (policy) drawn from a finite (possibly large) set,
and $Y_i\in\mathbb{R}$ is the outcome. Let $Y(t)$ denote the potential outcome under treatment $t\in\mathcal{T}$. For any $x$ and any pair $(t_1,t_0)$, we define the individualized uplift as the conditional treatment effect
\begin{equation}
\tau(x;t_1,t_0)
:= \mathbb{E}[Y(t_1)-Y(t_0)\mid X=x].
\label{eq:cate_def}
\end{equation}
At the unit level, the (unobservable) individual treatment effect is $\tau_i(t_1,t_0):=Y_i(t_1)-Y_i(t_0)$.
Throughout, $\tau(x;t_1,t_0)$ is the estimand (the standard CATE); in experiments we refer to the pointwise prediction
$\hat\tau(X_i;t_1,t_0)$ as an ``ITE'' score for brevity.

\begin{assumption} (SUTVA for RCT) Our main setting is randomized assignment.
We assume (i) consistency/SUTVA, (ii) randomized treatment assignment conditional on $X$,
and (iii) overlap. Under these conditions, $\tau(x;t_1,t_0)$ is identified from $\mathcal{D}$.
\end{assumption}

\subsection{Structured Treatment Representation}

\paragraph{Treatment-side contexts and atomic actions.}
We assume each policy operates through a finite set of treatment-side contexts
$\mathcal{S}$ and a finite set of atomic actions $\mathcal{A}$.
A policy $t\in\mathcal{T}$ is specified by a collection of context-dependent {sub-strategies}
$\{\Pi_t(\cdot\mid s)\}_{s\in\mathcal{S}}$, where each sub-strategy $\Pi_t(\cdot\mid s)$ is a distribution over actions in $\mathcal{A}$.
We treat $\{\Pi_t(\cdot\mid s)\}_{s\in\mathcal{S}}$ as the observable specification of a structured treatment.
We also allow nonnegative context weights $w(s)\ge 0$ that capture the relative exposure frequency or importance of context $s$
(e.g., segment prevalence), $w(\cdot)$ can be fixed from logs, domain knowledge or estimated.

\paragraph{Permutation-invariant treatment embedding.}
Let $\phi:\mathcal{S}\times\mathcal{A}\to\mathbb{R}^d$ be an embedding of a context--action atom
and let $\rho:\mathbb{R}^{d}\to\mathbb{R}^d$ be a learnable map.
We define an intermediate aggregated representation $h(t)\in\mathbb{R}^d$:
\begin{align}
z(t) &= \sum_{s\in\mathcal{S}}\sum_{a\in\mathcal{A}} w(s)\,\Pi_t(a\mid s)\,\phi(s,a),
\label{eq:z_def}\\
h(t) &= \rho\!\left(z(t)\right).
\label{eq:h_def}
\end{align}
This construction is invariant to any re-indexing or re-ordering of the underlying context--action components, and enables parameter sharing across policies through the shared atom embedding $\phi$.

\paragraph{Implementation of the treatment network.}
In practice, $\phi(s,a)$ is parameterized as a lookup table over $(s,a)$ (or a sum of separate embeddings for $s$ and $a$),
and $\rho(\cdot)$ is a small MLP (optionally with normalization).
The cost of computing $h(t)$ is $O(|\mathcal{S}||\mathcal{A}|d)$ per policy if done directly,
and can be reduced by caching $\{\phi(s,a)\}$ and exploiting sparsity in $\Pi_t(\cdot\mid s)$.

\subsection{Orthogonalized Factorized Uplift Model}

\paragraph{Model class.}
We consider Neyamen decomposition following \citep{kaddour2021causal}:
\begin{equation}
\begin{aligned}
Y &= m(X) + g(X)^\top\!\big(h(T)-e_h(X)\big) + \varepsilon,
\\
e_h(X)&:=\mathbb{E}[h(T)\mid X],
\label{eq:olr_model}
\end{aligned}
\end{equation}
where $m:\mathbb{R}^p\to\mathbb{R}$ is a baseline outcome regression model and $g:\mathbb{R}^p\to\mathbb{R}^d$ captures treatment effect heterogeneity.
Under \eqref{eq:olr_model}, the implied uplift is
\begin{equation}
\tau(x;t_1,t_0) = g(x)^\top\!\big(h(t_1)-h(t_0)\big).
\label{eq:cate_closed_form}
\end{equation}
We refer to this as ``factorized'' (or rank-$d$) since the effect depends on $(x,t)$ only through the inner product
between $g(x)$ and $h(t)$.

\paragraph{Learning objective.}
Let $m_\theta$, $g_\eta$, $\phi_\omega$, and $\rho_\nu$ parameterize $m,g,\phi,\rho$, and define $h_{\omega,\nu}$ by
\eqref{eq:z_def}--\eqref{eq:h_def}.
We estimate parameters by empirical risk minimization:
\begin{equation}
\begin{aligned}
\min_{\theta,\eta,\omega,\nu}\;
\frac{1}{n}\sum_{i=1}^n
\ell\!\left(
Y_i,\;
m_\theta(X_i)+g_\eta(X_i)^\top\!\big(h_{\omega,\nu}(T_i)-\hat e_h(X_i)\big)
\right) \\
+\lambda\,\Omega(\theta,\eta,\omega,\nu),
\label{eq:erm_obj}
\end{aligned}
\end{equation}
where $\ell$ is a prediction loss (e.g., squared loss or negative log-likelihood) and $\Omega$ is a regularizer.

\paragraph{Computing $\hat e_h(x)$.}\label{sec:estimation}
A key term in \eqref{eq:olr_model} is
$e_h(x)=\mathbb{E}[h(T)\mid X=x]$, which we approximate by $\hat e_h(x)$. Let the  propensity of treatment $t$ be $\pi_0(t\mid x):=\mathbb{P}(T=t\mid X=x)$. Two cases are considered. 1) \textbf{RCT Study.}
In the common case of complete randomization, treatment assignment is independent of covariates, i.e.,
$T \perp X$. Then $e_h(x)$ is a constant vector: $e_h(x)=\mathbb{E}[h(T)\mid X=x]=\mathbb{E}[h(T)].$ 

Accordingly, we estimate it by the empirical mean of $h(T)$ over the experiment data (or by a running mean in mini-batch SGD, as in Algorithm~\ref{alg:dsin}).
More generally, under stratified or covariate-adaptive randomization, $\pi_0(t\mid x)$ is known and may depend on $x$ through the stratum; in that case we set
$\hat e_h(x)=\sum_{t\in\mathcal{T}} \pi_0(t\mid x)\,h(t).$
2) \textbf{Observational Study.} When treatment assignment is not randomized, $e_h(x)$ generally varies with $x$ and $\pi_0(t\mid x)$ is unknown.
We estimate a propensity model $\hat\pi(t\mid x)\approx \pi_0(t\mid x)$ and plug it in:
$\hat e_h(x)=\sum_{t\in\mathcal{T}} \hat\pi(t\mid x)\,h(t).$
When $|\mathcal{T}|$ is large, the sum can be approximated by sampling treatments from $\hat\pi(\cdot\mid x)$.

\begin{algorithm}[t]
\caption{Permutation - invariant Orthogonal Uplift Learning (POUL)}
\label{alg:dsin}
\begin{algorithmic}[1]
\REQUIRE 
Full data $\mathcal{D}_{\mathrm{all}}=\{(X_i,Y_i)\}$; Experiment data $\mathcal{D}_{\mathrm{exp}}=\{(X_i,T_i,Y_i)\}$;
treatment-side contexts $\mathcal{S}$; atomic actions $\mathcal{A}$;
Policy specification $\{\Pi_t(a\mid s)\}_{t\in\mathcal{T},\,s\in\mathcal{S},\,a\in\mathcal{A}}$;
weights $w(\cdot)$; loss $\ell$; regularizer $\Omega$ and weight $\lambda$; learning rate $\gamma$

\STATE \textbf{Stage 1: fit outcome regression model.}
\STATE Initialize $m_\theta$ and train $m_\theta$ on $\mathcal{D}_{\mathrm{all}}$ by minimizing 
$\frac{1}{|\mathcal{D}_{\mathrm{all}}|}\sum \ell\!\big(m_\theta(X_i),Y_i\big)$
\STATE Freeze $\hat m(\cdot)\leftarrow m_\theta(\cdot)$

\STATE \textbf{Stage 2: train orthogonalized uplift model.}
\STATE Initialize $(\eta,\omega,\nu)$ for $g_\eta$, $\phi_\omega$, $\rho_\nu$
\STATE Initialize running mean $\bar h\in\mathbb{R}^d \leftarrow 0$, counter $c\leftarrow 0$
\FOR{each SGD step}
  \STATE Sample a mini-batch $\mathcal{B}\subset \mathcal{D}_{\mathrm{exp}}$
  \FOR{each $(X_i,T_i,Y_i)\in\mathcal{B}$}
    \STATE Compute treatment embedding (permutation-invariant aggregation):
    
      $z_i \leftarrow \sum_{s\in\mathcal{S}}\sum_{a\in\mathcal{A}} 
      w(s)\,\Pi_{T_i}(a\mid s)\,\phi_\omega(s,a),$\\$
      h_i \leftarrow \rho_\nu(z_i)$
    
    \STATE Compute two nets:
    $
      \hat m_i \leftarrow \hat m(X_i), g_i \leftarrow g_\eta(X_i)
    $
    \STATE Update $\hat e_h(x)$ (replace this term by the corresponding cases described in Sec.~\ref{sec:estimation}):
    $
      \hat e_h(X_i) \leftarrow \bar h
    $
    \STATE Predict:
    $
      \hat Y_i \leftarrow \hat m_i + g_i^\top\!\big(h_i-\hat e_h(X_i)\big)
    $
  \ENDFOR
  \STATE Get loss:
  $
    \mathcal{L}\leftarrow \frac{1}{|\mathcal{B}|}\sum_{(X_i,T_i,Y_i)\in\mathcal{B}} 
    \ell(\hat Y_i,Y_i) + \lambda\,\Omega(\eta,\omega,\nu)
  $
  \STATE SGD update: $(\eta,\omega,\nu)\leftarrow(\eta,\omega,\nu)-\gamma\nabla\mathcal{L}$
  \STATE Update $\bar h$ using current batch embeddings:\\
  $
    \bar h \leftarrow \frac{c\,\bar h + \sum_{(X_i,T_i,Y_i)\in\mathcal{B}} h_i}{c+|\mathcal{B}|},
    \qquad c \leftarrow c+|\mathcal{B}|
  $
\ENDFOR
\STATE \textbf{return} $\hat\tau(x;t_1,t_0)=g_{\hat\eta}(x)^\top\!\big(h_{\hat\omega,\hat\nu}(t_1)-h_{\hat\omega,\hat\nu}(t_0)\big)$
\end{algorithmic}
\end{algorithm}

\paragraph{Training Procedure}\label{subsec:estimation}
We train POUL via a two-stage orthogonal learning procedure (Algorithm~\ref{alg:dsin}). In stage 1 (lines 1-3), we fit a base predictor $m(x)\approx \mathbb{E}[Y\mid X=x]$ using all available data, which can be a DNN network, or Dragonet \citep{shi2019adapting}. This stage leverages the largest possible sample to obtain a stable estimate of the main outcome signal, independent of treatment structure.
After training, we treat $\hat m(\cdot)$ as fixed in the second stage. In stage 2 (lines 4-18), we learn the heterogeneous uplift component on the experiment population where policies are actually assigned and observed. Given a policy $t$, 
we compute its treatment embedding $h(t)$ by aggregating context--action atoms induced by its context-specific \emph{sub-strategies} $\{\Pi_t(\cdot\mid s)\}_{s\in\mathcal{S}}$ and context weights $w(s)$:
we first form the mixture representation $z(t)=\sum_{s\in\mathcal{S}}\sum_{a\in\mathcal{A}} w(s)\Pi_t(a\mid s)\phi(s,a)$ and then apply a learnable map $h(t)=\rho(z(t))$ (line~10).
The aggregation is \emph{permutation-invariant} to any re-indexing or re-ordering of contexts/actions, ensuring that two semantically identical policies (i.e., inducing the same mixture over atoms) have the same representation and preventing spurious dependence on arbitrary IDs or ordering. We fit the heterogeneity function $g(x)$ and the treatment embedding parameters by minimizing an orthogonalized loss based on the residualized outcome $Y-\hat m(X)$ and the residualized treatment representation $h(T)-\hat e_h$. In our main RCT setting with complete randomization, treatment assignment is independent of covariates ($T\perp X$), so we estimate $e_h$ by the empirical mean of $h(T)$ over the experiment population (or an equivalent running estimate during training).
This yields the final uplift estimator
$\hat\tau(x;t_1,t_0)=\hat g(x)^\top\big(\hat h(t_1)-\hat h(t_0)\big)$.

\subsection{Technical Challenges}\label{subsec:practical}

\textbf{Semantic invariance under re-indexing of structured policies.}
In practice, a policy is specified by a collection of context-specific \emph{sub-strategies} $\{\Pi_t(\cdot\mid s)\}_{s\in\mathcal{S}}$, where the indexing and ordering of contexts/actions can be arbitrary and may change across deployments.
A treatment representation that depends on such ordering can introduce spurious differences and hurt robustness when policies are re-serialized, updated, or composed from the same building blocks.
Our mixture-based construction in \eqref{eq:z_def} is permutation-invariant, so the embedding depends only on the induced mixture over context--action atoms, aligning representation with causal semantics.

\textbf{Long-tailed policies, limited overlap, and cold-start variants.}
In large combinatorial treatment spaces, full overlap over complete policies is unrealistic: many policies are rare, newly created, or observed only in a few strata.
Treating each policy as a categorical label leads to high-variance effect estimates and poor generalization in long-tailed regimes.
POUL mitigates this by sharing parameters across reusable context--action components through the shared atom embedding $\phi(s,a)$ and the invariant aggregation in \eqref{eq:z_def}, so estimation depends more on component coverage than repeated exposure to identical full policies.
This parameter sharing enables rare or new policies to borrow statistical strength from frequent ones and improves stability under small policy perturbations.
In practice, we further improve stability using standard techniques such as propensity clipping (for observational data), restricting evaluation to regions with sufficient support, and regularizing the treatment embedding to avoid extrapolation from extremely sparse components.

\textbf{Noisy weights and scalable computation of mixture embeddings.}
Context weights $w(s)$ may be noisy proxies of exposure frequency or importance, and directly computing \eqref{eq:z_def} scales with $|\mathcal{S}||\mathcal{A}|$ per policy.
In practice, $w(\cdot)$ can be estimated from logs or domain knowledge, and we empirically assess robustness under imperfect weights.
For efficiency, we cache $\phi(s,a)$ and exploit sparsity in $\Pi_t(\cdot\mid s)$ to reduce computation, making training feasible at scale.
\section{Theoretical Analysis}
\label{theory}

Combinatorial policies induce high-dimensional treatments: each $t\in\mathcal{T}$ specifies a collection of
context-wise action distributions $\{\Pi_t(\cdot\mid s)\}_{s\in\mathcal{S}}$.
Our analysis answers three questions that are fundamental for making causal inference tractable in this regime:
\begin{enumerate}
\item \textbf{Expressiveness (representation).}
Can the proposed permutation-invariant embedding $h(t)$ capture a broad class of policy effects that depend on $t$
only through its induced context--action mixture? (Theorem~\ref{thm:expressiveness}.)
\item \textbf{Orthogonal robustness (estimation).}
When uplift is learned via the orthogonalized low-rank model \eqref{eq:olr_model},
how do nuisance estimation errors (e.g., baseline and embedding propensity) enter the final CATE error?
We show a \emph{second-order} (product-form) dependence typical of double/debiased ML. (Theorem~\ref{thm:orth_robust}.)
\item \textbf{Stability (generalization).}
If a new/rare policy $t'$ differs only slightly in its specification $\Pi_{t'}(\cdot\mid s)$,
does the learned embedding (and hence predicted uplift) vary smoothly rather than jumping arbitrarily?
We prove a Lipschitz-type stability bound in a natural distance over policy specifications. (Proposition~\ref{prop:stability}.)
\end{enumerate}
Together, these results formalize why \emph{(i)} the embedding is not ad-hoc,
\emph{(ii)} orthogonalization protects uplift learning from nuisance errors,
and \emph{(iii)} the method can generalize across combinatorial treatments beyond those frequently observed. For preparation, we introduce {policy-induced mixtures as a sufficient causal interface}.
Recall the context weights $w(s)\ge 0$ with $\sum_{s\in\mathcal{S}} w(s)=1$ and define
\begin{equation}
\alpha_t(s,a) \;:=\; w(s)\,\Pi_t(a\mid s),\qquad (s,a)\in\mathcal{S}\times\mathcal{A}.
\label{eq:alpha_def}
\end{equation}
The collection $\{\alpha_t(s,a)\}_{s,a}$ is a probability mass function over $\mathcal{S}\times\mathcal{A}$.
We denote by $\mu_t$ the corresponding discrete probability measure on $\mathcal{S}\times\mathcal{A}$, i.e.,
$\mu_t(\{(s,a)\})=\alpha_t(s,a)$.

\begin{assumption}[Policy-induced mixture sufficiency]
\label{ass:mixture_suff}
There exist measurable functions $m_0:\mathbb{R}^p\to\mathbb{R}$,
$g_0:\mathbb{R}^p\to\mathbb{R}^d$,
and a (possibly unknown) functional $F$ mapping measures on $\mathcal{S}\times\mathcal{A}$ to $\mathbb{R}^d$
such that for all $x\in\mathbb{R}^p$ and all $t\in\mathcal{T}$,
\begin{equation}
\mathbb{E}\!\left[Y(t)\mid X=x\right]
\;=\;
m_0(x) + g_0(x)^\top F(\mu_t).
\label{eq:mixture_suff}
\end{equation}
Moreover, $F$ is continuous on $\{\mu_t:t\in\mathcal{T}\}$ under the $\ell_1$/total-variation topology, i.e.,
if $\|\mu-\mu'\|_{1}\to 0$ then $\|F(\mu)-F(\mu')\|\to 0$.
\end{assumption}

Assumption~\ref{ass:mixture_suff} says that a policy affects outcomes only through the \emph{mixture it induces}
over context--action atoms. This matches the reality of many platform interventions.
For example, in ride-hailing marketplaces, a policy may specify (for each spatiotemporal context $s$)
a distribution over incentive actions $a$ (e.g., bonus levels or pricing multipliers).
The realized user experience aggregates over contexts according to exposure weights $w(s)$;
thus it is the induced mixture $\mu_t$---not an arbitrary index ordering of contexts/actions---that
governs the causal response. Also, Assumption~\ref{ass:mixture_suff} implies our orthogonalized model class \eqref{eq:olr_model} is well-specified as follows
(up to reparametrization).

\begin{proposition}[Well-specified orthogonalized low-rank representation]
\label{prop:olr_wellspec}
Under Assumption~\ref{ass:mixture_suff} and consistency $Y=Y(T)$,
define the \emph{oracle} treatment representation $h_0(t):=F(\mu_t)$ and
$e_0(x):=\mathbb{E}[h_0(T)\mid X=x]$.
Then there exists a function $m(\cdot)$ such that the conditional mean satisfies 
\begin{equation}
\begin{aligned}
\mathbb{E}[Y\mid X,T]
&=
m(X) + g_0(X)^\top\!\big(h_0(T)-e_0(X)\big),
\\
m(x)&:=m_0(x)+g_0(x)^\top e_0(x),
\label{eq:olr_wellspec}
\end{aligned}
\end{equation}
and the CATE defined in \eqref{eq:cate_def} obeys
\begin{equation}
\tau(x;t_1,t_0) = g_0(x)^\top\!\big(h_0(t_1)-h_0(t_0)\big).
\label{eq:cate_oracle}
\end{equation}
\end{proposition}

\noindent
Proposition~\ref{prop:olr_wellspec} clarifies a key boundary:
the orthogonalized decomposition \eqref{eq:olr_model} is \emph{not} an additional restriction beyond
\eqref{eq:mixture_suff}; it is a reparametrization that is convenient for debiased/orthogonal learning.
In particular, the causal estimand depends on $(g_0,h_0)$ but not on the nuisance pair $(m,e_0)$.

\subsection{Expressiveness of the Permutation-Invariant Policy Embedding}
\label{sec:expressive}

We now show that the embedding family \eqref{eq:z_def}--\eqref{eq:h_def} is expressive enough to approximate
any continuous functional $F(\mu_t)$ of the policy-induced mixture.
Given maps $\phi:\mathcal{S}\times\mathcal{A}\to\mathbb{R}^r$ and $\rho:\mathbb{R}^r\to\mathbb{R}^d$, define
\begin{equation}
z_{\phi}(t)=\sum_{s\in\mathcal{S}}\sum_{a\in\mathcal{A}} \alpha_t(s,a)\,\phi(s,a),~
h_{\phi,\rho}(t)=\rho\!\left(z_{\phi}(t)\right).
\label{eq:embedding_class}
\end{equation}
This construction is permutation-invariant to any re-indexing of $\mathcal{S}\times\mathcal{A}$
because it depends on $\{(s,a,\alpha_t(s,a))\}$ only through a weighted sum.

\begin{lemma}[Expressiveness of permutation-invariant embeddings]
\label{thm:expressiveness}
Under Assumption~\ref{ass:mixture_suff}, conditioning that $\mathcal{S}$ and $\mathcal{A}$ are finite and let
$\mathcal{M}:=\{\mu_t:t\in\mathcal{T}\}\subset\Delta(\mathcal{S}\times\mathcal{A})$.
Let $F:\mathcal{M}\to\mathbb{R}^d$ be continuous.
Then for any $\varepsilon>0$, there exist an integer $r$ and functions
$\phi:\mathcal{S}\times\mathcal{A}\to\mathbb{R}^r$ and $\rho:\mathbb{R}^r\to\mathbb{R}^d$
(such that $\rho$ can be implemented by a standard universal approximator, e.g., a ReLU MLP)
for which
\begin{equation}
\sup_{t\in\mathcal{T}}
\big\|
F(\mu_t) - h_{\phi,\rho}(t)
\big\|
\;\le\;
\varepsilon.
\label{eq:expressiveness_bound}
\end{equation}
\end{lemma}

Theorem~\ref{thm:expressiveness} justifies the representation choice in \eqref{eq:z_def}--\eqref{eq:h_def}:
\emph{if} the causal response of a policy depends on $t$ only through the induced mixture over context--action atoms,
\emph{then} our embedding family can approximate that dependence arbitrarily well.
In other words, the proposed $h(t)$ is a principled ``policy featurization'' rather than a heuristic encoding.
For example, consider a policy that assigns high incentives in peak Beijing and zero incentives in off-peak Shanghai. Two engineering implementations may store this policy using different context orderings or action identifiers. A categorical policy-ID representation would treat them as unrelated treatments, while our permutation-invariant embedding maps both specifications to the same representation, as they induce the same context–action mixture. Consequently, the implied uplift is identical, as required by causal semantics.

Noteworthy, such a {representation capacity} statement does \emph{not} claim:
(i) causal identification without Assumption~\ref{ass:mixture_suff} and the usual causal conditions (SUTVA/overlap),
(ii) that the required dimension $r$ is always small,
or (iii) that optimization will always find the approximating parameters.
These are separate issues (identification, statistical rates, and optimization) from expressiveness. Also, the skepticism arises that \emph{``If $\mathcal{S}\times\mathcal{A}$ is finite, isn't this just a fancy way of learning $F(\alpha)$ with an MLP?
Why do we need the embedding structure at all?''} It highlights what our contribution \emph{really} is:
we are not merely encoding $\alpha_t$; we are enforcing a \emph{compositional inductive bias}
(shared atom embeddings $\phi(s,a)$ and permutation-invariant aggregation).
This bias is what enables (i) parameter sharing across policies,
(ii) meaningful continuity in policy space (formalized in Proposition~\ref{prop:stability}),
and (iii) practical scaling when $\mathcal{T}$ is huge but $\mathcal{S}\times\mathcal{A}$ has structure.

Having established that the embedding family can represent the causal policy effect under
Assumption~\ref{ass:mixture_suff}, we now turn to estimation: how does orthogonalization
protect uplift learning from nuisance estimation errors?

\subsection{Orthogonal Robustness}
\label{sec:orth_robust}

Let $\pi_0(t\mid x):=\mathbb{P}(T=t\mid X=x)$ be the (known or estimable) assignment/propensity.
Let $h_0$ be the oracle representation from Proposition~\ref{prop:olr_wellspec} and define
$e_0(x):=\mathbb{E}[h_0(T)\mid X=x]=\sum_{t\in\mathcal{T}} \pi_0(t\mid x)\,h_0(t)$.
The orthogonalized regression form \eqref{eq:olr_model} suggests learning $g(\cdot)$ and $h(\cdot)$ from
residualized signals $(Y-m(X))$ and $(h(T)-e_h(X))$.
For theory, we analyze the standard cross-fitted variant that isolates nuisance estimation.
(In experimental design with known $\pi_0$, the propensity step is omitted; the same analysis applies with smaller nuisance error.)
Let $\{I_k\}_{k=1}^K$ be a partition of $\{1,\dots,n\}$.

Define the \emph{Neyman orthogonality} score
$\psi(W; g,h,m,e)
:=
\Big(Y-m(X)-g(X)^\top(h(T)-e(X))\Big)\,(h(T)-e(X))$ where $W:=(X,T,Y).$ We recall its first-order invariance as follows.

\begin{lemma}[Neyman orthogonality w.r.t.\ nuisance $(m,e)$]
\label{lem:neyman_orth}
Assume the model \eqref{eq:olr_model} holds with some $(m_0,g_0,h_0,e_0)$ and
$\mathbb{E}[\varepsilon\mid X,T]=0$, where $e_0(x)=\mathbb{E}[h_0(T)\mid X=x]$.
Then
$\mathbb{E}\big[\psi(W; g_0,h_0,m_0,e_0)\big]=0,$
and the Gateaux derivative of $\mathbb{E}[\psi(W; g_0,h_0,m,e)]$ in the direction of perturbations
$(\delta m,\delta e)$ vanishes at $(m,e)=(m_0,e_0)$:
$\left.\frac{d}{dr}\right|_{r=0}
\mathbb{E}\Big[\psi(W; g_0,h_0,m_0+r\,\delta m,\,e_0+r\,\delta e)\Big]
=0.$
\end{lemma}

Lemma~\ref{lem:neyman_orth} is the formal statement of
``nuisance errors do not affect the target to first order.''
It is exactly the property exploited by debiased/double machine learning. Additionally, considering its second-order, we now state a pointwise robustness guarantee in the spirit of DML error expansions.
Below, for a (vector-valued) function $f$ we use $\|f\|_{2}:=(\mathbb{E}\|f(W)\|^2)^{1/2}$,
and for $h(\cdot)$ we use $\|h-\tilde h\|_{\infty}:=\sup_{t\in\mathcal{T}}\|h(t)-\tilde h(t)\|$.

\begin{assumption}[Regularity for orthogonal robustness]
\label{ass:regularity_orth}
\leavevmode
\begin{enumerate}
\item (Bounded moments) $\mathbb{E}[Y^2]<\infty$ and $\sup_{t\in\mathcal{T}}\|h_0(t)\|<\infty$,
and $\mathbb{E}[\|g_0(X)\|^2]<\infty$.
\item (Overlap) There exists $c>0$ s.t.\ $\pi_0(t\mid x)\ge c$ for all $t\in\mathcal{T}$ and a.e.\ $x$.
\item (Cross-fitting) Nuisance estimates $(\hat m^{(-k)},\hat\pi^{(-k)})$ are fit on data independent of fold $I_k$.
\item (Estimation errors) Define
$\delta_m:=\|\hat m-m_0\|_2$,
$\delta_e:=\|\hat e-e_0\|_2$,
$\delta_g(x):=\|\hat g(x)-g_0(x)\|$,
and $\delta_h:=\|\hat h-h_0\|_{\infty}$.
Assume these quantities are $o_p(1)$.
\end{enumerate}
\end{assumption}

\begin{theorem}[Orthogonal robustness: nuisance errors enter at second order]
\label{thm:orth_robust}
Suppose Assumption~\ref{ass:regularity_orth} holds.
Let $\hat\tau(x;t_1,t_0)=\hat g(x)^\top(\hat h(t_1)-\hat h(t_0))$ be the estimator returned by
Algorithm~\ref{alg:dsin}. Then for any fixed $x$ and $(t_1,t_0)$,
$\hat\tau(x;t_1,t_0)-\tau(x;t_1,t_0)=\textbf{error(g)+\textbf{error(h)}+error(second order)}$, where
$\textbf{error(g)}:=(\hat g(x)-g_0(x))^\top\!\big(h_0(t_1)-h_0(t_0)\big)$, $\textbf{error(h)}:=g_0(x)^\top\!\Big((\hat h-h_0)(t_1)-(\hat h-h_0)(t_0)\Big)$ and $\textbf{error(second order)}:= O_p\Big( (\delta_m+\delta_h+\sup_{x}\delta_g(x))\,\delta_e \Big)$. In particular, the nuisance pair $(m,e)$ affects the CATE only through the product term
$(\delta_m+\delta_h+\sup_x\delta_g(x))\delta_e$, i.e., there is no first-order (linear) sensitivity to $\delta_m$ or $\delta_e$.
\end{theorem}

Theorem~\ref{thm:orth_robust} explains why orthogonalization is essential in large-scale policy learning:
even when baseline regression $m$ and embedding propensity $e$ are estimated with complex ML models,
their errors do not linearly bias the uplift estimator; they only contribute as a product (second-order term).
This is precisely the regime where DML is powerful: nuisance learners can be high-capacity without destabilizing the target.

As a concrete special case, if $\delta_m=o_p(n^{-1/4})$ and $\delta_e=o_p(n^{-1/4})$,
then $\delta_m\delta_e=o_p(n^{-1/2})$ and the nuisance contribution becomes negligible compared with typical
$1/\sqrt{n}$ sampling fluctuations (mirroring classical DML conditions).

Expressiveness (Theorem~\ref{thm:expressiveness}) and orthogonal robustness (Theorem~\ref{thm:orth_robust})
address \emph{what} we can represent and \emph{how} we can estimate it robustly.
We next address \emph{why} this representation supports generalization across rare/unseen combinatorial policies:
small changes in policy specification lead to small changes in $h(t)$ and thus in uplift.

\subsection{Stability Under Policy Perturbations}
\label{sec:stability}

\paragraph{A natural metric on policy specifications.}
Define the weighted $\ell_1$ distance between policies by
\begin{equation}
d_{\Pi}(t,t')
:=
\sum_{s\in\mathcal{S}} w(s)\,\big\|\Pi_t(\cdot\mid s)-\Pi_{t'}(\cdot\mid s)\big\|_{1}.
\label{eq:policy_distance}
\end{equation}
This distance measures how much the context-wise action distributions change, averaged by context exposure $w(s)$.
Note that $d_{\Pi}(t,t')=\|\mu_t-\mu_{t'}\|_{1}$ under the identification $\mu_t(s,a)=w(s)\Pi_t(a\mid s)$.

\begin{proposition}[Stability of the embedding and uplift]
\label{prop:stability}
Assume $\rho$ is $L_\rho$-Lipschitz, i.e., $\|\rho(u)-\rho(v)\|\le L_\rho\|u-v\|$ for all $u,v$,
and the atom embedding is bounded: $\sup_{(s,a)}\|\phi(s,a)\|\le B$.
Then for any $t,t'\in\mathcal{T}$,
\begin{equation}
\|h(t)-h(t')\|
\;\le\;
L_\rho\,\|z(t)-z(t')\|
\;\le\;
L_\rho\,B\, d_{\Pi}(t,t').
\label{eq:embedding_stability}
\end{equation}
If moreover $\|g_0(x)\|\le G$ for the $x$ of interest, then the (oracle) uplift satisfies 
$
\big|\tau(x;t,t')\big|
=
\big|g_0(x)^\top(h_0(t)-h_0(t'))\big|
\;\le\;
G\,\|h_0(t)-h_0(t')\|
\;\le\;
G\,L_\rho\,B\, d_{\Pi}(t,t').
$
\end{proposition}

Proposition~\ref{prop:stability} turns an empirical phenomenon into a principled statement:
in combinatorial treatment problems, \emph{treatments are not isolated labels};
they come with a structured specification $\Pi_t$.
Our embedding respects this structure and yields a continuity guarantee:
if a new policy $t'$ differs from a known policy $t$ only slightly (in $d_\Pi$),
then both the learned representation and the implied uplift change only slightly.
This provides a theoretical explanation for why the method can extrapolate to rare/unseen policies that are
\emph{nearby in specification space}.

{A common misconception} is that \emph{``Any neural network can be made Lipschitz; isn't this trivial?''} The nontrivial part is \emph{which input space the Lipschitzness is defined on}.
Here it is defined on the \emph{policy specification distance} $d_\Pi$,
which is meaningful for combinatorial policies and unavailable to methods that treat $t$ as a categorical ID.
A categorical-ID model has no reason to behave smoothly as a policy is perturbed,
because ``nearby'' has no definition in that representation.

Concluding above, Assumption~\ref{ass:mixture_suff} provides a structural interface: policies act through induced context--action mixtures.
Theorem~\ref{thm:expressiveness} then guarantees that our permutation-invariant embedding family is rich enough to
approximate the corresponding causal functional $F(\mu_t)$.
Given such an embedding, Theorem~\ref{thm:orth_robust} shows that orthogonalized learning yields a CATE estimator
whose dependence on nuisance errors is second order, aligning with the core promise of double machine learning.
Finally, Proposition~\ref{prop:stability} explains how the representation supports generalization across
combinatorial policies: small specification perturbations translate to controlled changes in embeddings and uplift.

\section{Experiments}\label{exp}
In this section, we conduct comprehensive experiments to validate the effectiveness of our proposed POUL framework. Through rigorous evaluations on large-scale industrial datasets, we demonstrate that POUL achieves superior Individual Treatment Effect (ITE) estimation and enhanced robustness against selection bias compared to competitive baselines. All comprehensive implementation details and experimental setups are provided in Appendix~\ref{sec:exp_details}.

\begin{table*}[t]
\centering
\caption{Combined Estimation Performance. Comparison of AUUC ($\uparrow$) and MAPE ($\downarrow$) for both GMV (Panel A) and Cost (Panel B) estimation tasks. The results are reported across Global ($\mathcal{D}_{G}$) and Core ($\mathcal{D}_{C}$) populations. Best results are \textbf{bolded}.}
\label{tab:combined_performance_wide}
\begin{small}
\begin{sc}
\begin{tabular}{l|cccc|cccc}
\toprule
Dataset & \multicolumn{4}{c|}{Global Population ($\mathcal{D}_{G}$)} & \multicolumn{4}{c}{Core Population ($\mathcal{D}_{C}$)} \\
\midrule
\multirow{2}{*}{Metrics} & \multicolumn{2}{c}{Treatment 1 ($T_1$)} & \multicolumn{2}{c|}{Treatment 2 ($T_2$)} & \multicolumn{2}{c}{Treatment 1 ($T_1$)} & \multicolumn{2}{c}{Treatment 2 ($T_2$)} \\
& AUUC $\uparrow$ & MAPE $\downarrow$ & AUUC $\uparrow$ & MAPE $\downarrow$ & AUUC $\uparrow$ & MAPE $\downarrow$ & AUUC $\uparrow$ & MAPE $\downarrow$ \\
\midrule
\multicolumn{9}{c}{\textbf{Panel A: GMV (LTV) Estimation}} \\
\midrule
Random & 0.500 & -- & 0.500 & -- & 0.500 & -- & 0.500 & -- \\
Causal Forest & 0.611 & 1.000 & 0.632 & 1.850 & 0.536 & 1.086 & 0.536 & 0.836 \\
DESCN  & 0.514 & 2.156 & 0.597 & \textbf{0.977} & 0.596 & 1.256 & 0.591 & 1.167 \\
DRCFR  & 0.660 & \textbf{0.779} & 0.612 & 3.389 & 0.622 & \textbf{0.888} & 0.566 & \textbf{0.713} \\
T-Learner    & 0.581 & 3.310 & \textbf{0.639} & 2.711 & 0.585 & 2.487 & 0.563 & 1.426 \\
POUL (Ours)   & \textbf{0.826} & 1.333 & 0.635 & 1.386 & \textbf{0.792} & 1.570 & \textbf{0.690} & 2.023 \\
\midrule
\multicolumn{9}{c}{\textbf{Panel B: Cost Estimation}} \\
\midrule
Random & 0.500 & -- & 0.500 & -- & 0.500 & -- & 0.500 & -- \\
Causal Forest & 0.692 & 0.792 & 0.697 & 0.881 & 0.576 & 0.946 & 0.596 & 0.969 \\
DESCN  & 0.735 & \textbf{0.113} & 0.777 & 0.409 & 0.608 & \textbf{0.103} & 0.641 & 0.456 \\
DRCFR  & 0.715 & 0.364 & 0.753 & 0.607 & 0.586 & 0.177 & 0.641 & \textbf{0.344} \\
T-Learner    & \textbf{0.739} & 0.221 & 0.777 & \textbf{0.316} & 0.616 & 0.191 & 0.645 & 0.445 \\
POUL (Ours)   & 0.735 & 0.838 & \textbf{0.917} & 0.755 & \textbf{0.736} & 0.844 & \textbf{0.916} & 0.769 \\
\bottomrule
\end{tabular}
\end{sc}
\end{small}
\end{table*}

\subsection{Experimental Design}
\paragraph{Datasets.} We evaluate our framework on a large-scale ride-hailing dataset with user features ($p=723$). The analysis spans two population scopes: a \textit{Global} pool of 3,820,637 users and a \textit{Core Eligible} subset of 897,946 users satisfying the inclusion criteria for service class upgrades. The treatment space comprises $T_1$ (\textit{Economy-to-Express}) and a composite $T_2$ (\textit{Discount-to-Economy} + \textit{Economy-to-Express}). Crucially, treatment triggering is contingent upon the intersection of user opt-ins and specific operational scenario constraints. Upon satisfying these conditions, valid \textit{Economy-to-Express} requests are stochastically assigned to $T_1$ or $T_2$, whereas \textit{Discount-to-Economy} requests exclusively trigger $T_2$, introducing structural dependencies between user intent, environmental context, and treatment assignments. Furthermore, we simultaneously model two distinct outcomes. \textit{Monthly Gross Merchandise Value} ($GMV$) serves as an observable monthly surrogate for the user's \textit{Lifetime Value} ($LTV$), while \textit{Monthly Cost} ($Cost$) explicitly quantifies the monthly operational expenditure incurred by the class upgrade treatment.

\paragraph{Baselines.} We compare our proposed method against representative baselines, including the non-parametric \textbf{Causal Forest}~\cite{wager2018estimation}, and three deep learning methods implemented with an \textbf{MLP} backbone: \textbf{T-Learner}~\citep{kunzel2019metalearners} (implemented via distinct DNN heads to separately evaluate and optimize the regression loss for the treatment and control groups), \textbf{DESCN}~\citep{zhong2022descn}, and \textbf{DRCFR}~\citep{cheng2022learning}.

\paragraph{Evaluation Metrics. }
To evaluate the effectiveness of POUL, we employ two standard metrics on a randomized test set: normalized \textbf{Area Under the Uplift Curve (AUUC)} to assess uplift ranking capability, and \textbf{Mean Absolute Percentage Error (MAPE)} to quantify ITE estimation accuracy. We refer details to Appendix~\ref{sec:appendix_metrics}.

\subsection{Experiment results}
We optimize hyperparameters via grid search and obtain the following optimal results in Table~\ref{tab:combined_performance_wide}. For further details, please refer to Appendix~\ref{sec:results_details}. As shown in Table~\ref{tab:combined_performance_wide}, POUL achieves state-of-the-art performance across both Global ($\mathcal{D}_{G}$) and Core ($\mathcal{D}_{C}$) populations, surpassing baselines by \textbf{25.2\%} in the high-variance GMV task ($T_1$).
Crucially, our model exhibits dominant superiority in the combinatorial strategy setting ($T_2$), achieving an AUUC of \textbf{0.916} in Cost estimation and exceeding baseline by \textbf{42.0\%} in the Core population.
This empirical gap directly corroborates our theoretical analysis: unlike traditional baselines, POUL represents the policy as an \textit{induced mixture over context--action atoms}, guaranteeing stability under perturbations and enabling the capture of complex policy interactions in $T_2$.

\section{Conclusion}\label{conclu}

This work studies uplift estimation for combinatorial policies, representing treatments via the mixtures they induce over context--action components. We establish the expressiveness of permutation-invariant policy embeddings, the orthogonal robustness of the resulting low-rank estimator, and stability under policy perturbations. These results confirm that respecting policy structure and invariance is crucial for reliable causal learning in large, evolving treatment spaces.

Future work should extend the framework both theoretically and practically. Theoretical analysis could address settings with unobserved confounding, non-stationarity, or evolving context--action sets. Methodologically, integration with meta-learning, Bayesian inference, or sequential decision-making models could enhance adaptability and uncertainty quantification. Applying the approach to broader domains---such as dynamic pricing or personalized healthcare---and developing standardized benchmarks would further validate its utility and robustness.

In all, combining causality with structured representation learning offers a principled path toward generalizable and interpretable decisions. Further exploration of its connections to causal discovery, programmatic policies, partial identification, and reinforcement learning represents a promising and crucial frontier for intelligent decision-making.

\newpage

\bibliography{example_paper}
\bibliographystyle{conf2026}

\clearpage
\appendix
\onecolumn

\section{Related Work}
\label{related}

\textbf{CATE and uplift estimation.}
A large literature studies heterogeneous treatment effect estimation under the potential outcomes framework, with both classical and modern ML approaches
\citep{rubin1974estimating, imbens2015causal, athey2016recursive, wager2018estimation}.
For binary treatments, meta-learners such as the T-/S-/X-learners provide practical plug-in strategies and are widely used in industrial uplift modeling
\citep{kunzel2019metalearners}.
In parallel, uplift modeling has developed evaluation protocols centered around incremental gain and Qini-type curves, as well as specialized tree learners for direct uplift estimation
\citep{radcliffe2011real, rzepakowski2012decision}.
Despite substantial progress, most of these methods implicitly assume a \emph{small} and \emph{fixed} treatment set.
When the treatment space becomes large and long-tailed—as in combinatorial policies—treating each treatment as an unrelated category leads to poor sample efficiency and unstable ranking, because there is no mechanism to share information across related interventions.

\textbf{Orthogonalization, doubly robust learning, and debiased ML.}
Orthogonal / doubly robust ideas are a cornerstone of modern causal estimation, mitigating sensitivity to nuisance estimation through Neyman orthogonality and cross-fitting
\citep{robins1994estimation, vanderlaan2011targeted, chernozhukov2018double, kennedy2020optimal}.
These works provide a principled recipe: learn flexible nuisance components (outcome regression, propensity) while guaranteeing that their errors enter the target estimator only at second order.
However, the majority of the DML/TMLE literature treats the treatment as low-dimensional (binary or small discrete), and the nuisance objects are designed accordingly.
In our setting, the primary bottleneck is not only nuisance estimation, but also \emph{how to represent a high-cardinality, structured policy} in a way that preserves causal semantics.
Our contribution is complementary: we retain the orthogonal robustness advantages of DML-style learning \citep{chernozhukov2018double} while introducing a structured, permutation-invariant treatment representation that enables parameter sharing and stable generalization across a combinatorial policy space.

\textbf{Representation learning for treatment effects.}
Deep representation learning has been widely adopted for CATE/ITE estimation, often aiming to reduce distribution shift between treated and control groups in observational data.
Notable examples include TARNet and its balancing variants \citep{shalit2017estimating}, as well as architectures that explicitly incorporate propensity information such as DragonNet and targeted regularization
\citep{shi2019adapting}.
Extensions to multi-treatment and dose-response settings have also been developed, e.g., DRNet for multiple treatments with continuous dosage \citep{schwab2019learning}.
These methods primarily focus on learning a good representation of \emph{unit covariates} $X$ (and nuisance functions of $X$) to improve counterfactual generalization.
By contrast, our focus is orthogonal: we study how to represent \emph{structured treatments} (policies) themselves.
In particular, even with randomized assignment (where ignorability is ensured), naive treatment encodings can destroy equivalences between policies and prevent information sharing across related policies, leading to systematic instability in long-tailed regimes.
Our permutation-invariant treatment embedding targets precisely this failure mode.\\

\textbf{Structured and compositional treatments.}

{There is increasing interest in causal estimation under nonstandard treatments, including multiple treatments, continuous treatments, and structured interventions \citep{kunzel2019metalearners,schwab2019learning}.
In many real systems, a ``treatment'' is naturally a \textit{policy} mapping treatment-side contexts to action distributions, inducing a compositional structure.
Most existing multi-treatment architectures model treatments via one-hot or learned embeddings tied to treatment IDs.
While recent approaches like \textbf{EFIN} \citep{liu2023explicit} mitigate this by explicitly encoding treatment attributes (e.g., coupon values) to capture inter-treatment correlations, they fundamentally rely on fixed feature interactions rather than enforcing invariances implied by policy semantics.
As a result, two implementations of the same policy (differing only by re-indexing or ordering) may still be treated as unrelated, and nearby policies in specification space need not be nearby in representation space.
Our work addresses this gap by representing a policy through its induced mixture over context--action atoms and by proving stability under small policy perturbations, which is absent from categorical-ID and feature-interaction paradigms.}

\textbf{Permutation invariance and symmetry-aware learning.}
Permutation-invariant architectures such as Deep Sets \citep{zaheer2017deepsets} and more general symmetry-aware frameworks \citep{kondor2018generalization} provide a powerful inductive bias when the target depends only on a multiset of components.
While such invariances are standard in set-structured prediction, their role in causal estimation with structured treatments has not been systematically developed.
Our treatment representation can be viewed as a causal instantiation of this principle: the treatment embedding is a function of the multiset of context--action components weighted by policy-induced mixture probabilities, and is therefore invariant to re-indexing and ordering.
Crucially, this is not merely an architectural convenience: the induced invariance aligns the statistical representation with causal semantics and enables both theoretical stability guarantees and practical generalization to rare or unseen policy variants.

In short, prior work has made significant progress on (i) CATE estimation for low-cardinality treatments, (ii) orthogonal robustness to nuisance estimation, and (iii) representation learning over covariates.
Nevertheless, under \emph{combinatorial policy treatments} with frequent iteration and long-tailed exposure, existing methods that rely on categorical treatment IDs remain systematically biased in the \emph{representation} step: they fail to respect policy equivalences and provide no meaningful notion of proximity in policy space.
We depart from this paradigm by combining (A) a permutation-invariant, compositional treatment embedding with (B) an orthogonalized uplift objective, thereby obtaining a framework that (1) contains standard DML-style robustness as a special case, (2) generalizes across structured treatments through parameter sharing, and (3) admits explicit expressiveness and stability guarantees tailored to policy perturbations.

\section{Proofs and Additional Technical Details}
\label{app:proofs}

\subsection{Notation and basic identities}
\label{app:notation}

Recall that each policy $t\in\mathcal{T}$ is specified by context-wise action distributions
$\{\Pi_t(\cdot\mid s)\}_{s\in\mathcal{S}}$ over a finite action set $\mathcal{A}$ and finite context set $\mathcal{S}$.
Let $w:\mathcal{S}\to[0,\infty)$ be context weights with $\sum_{s\in\mathcal{S}}w(s)=1$ and define the induced atom
weights
\[
\alpha_t(s,a)\;:=\;w(s)\Pi_t(a\mid s),\qquad (s,a)\in\mathcal{S}\times\mathcal{A},
\]
and the induced discrete measure $\mu_t$ on $\mathcal{S}\times\mathcal{A}$ by $\mu_t(\{(s,a)\})=\alpha_t(s,a)$.

Throughout the appendix, for a scalar (or vector) random variable $Z$ we write
$\|Z\|_2 := (\mathbb{E}\|Z\|^2)^{1/2}$.
For a function $h:\mathcal{T}\to\mathbb{R}^d$, we use
$\|h\|_{\infty} := \sup_{t\in\mathcal{T}}\|h(t)\|$ and
$\|h-\tilde h\|_{\infty} := \sup_{t\in\mathcal{T}}\|h(t)-\tilde h(t)\|$.
For a function $g:\mathbb{R}^p\to\mathbb{R}^d$, we will occasionally write
$\|g-\tilde g\|_{\infty,x} := \sup_{x}\|g(x)-\tilde g(x)\|$ when the supremum is taken over the domain of interest.

A basic identity that will be used repeatedly is the ``centering'' property
\begin{equation}
\mathbb{E}\!\left[h_0(T)-e_0(X)\mid X\right]=0,\qquad e_0(X):=\mathbb{E}[h_0(T)\mid X].
\label{eq:centering_identity}
\end{equation}
This is the main mechanism behind Neyman orthogonality and the second-order nature of nuisance effects.

\subsection{Proof of Proposition~\ref{prop:olr_wellspec}}

\begin{proof}[Proof of Proposition~\ref{prop:olr_wellspec}]
Fix any $x\in\mathbb{R}^p$ and $t\in\mathcal{T}$.
By Assumption~\ref{ass:mixture_suff}, there exist measurable $m_0$ and $g_0$ and a functional $F$ such that
\[
\mathbb{E}[Y(t)\mid X=x]
=
m_0(x)+g_0(x)^\top F(\mu_t).
\]
Define $h_0(t):=F(\mu_t)$.
Under consistency $Y=Y(T)$ and randomized assignment (so that conditioning on $(X,T)$ selects the corresponding
potential outcome), we have
\[
\mathbb{E}[Y\mid X=x,T=t]
=
\mathbb{E}[Y(t)\mid X=x]
=
m_0(x)+g_0(x)^\top h_0(t).
\]
Let $e_0(x):=\mathbb{E}[h_0(T)\mid X=x]$ and define
$m(x):=m_0(x)+g_0(x)^\top e_0(x)$.
Then, for any $(x,t)$,
\begin{align*}
m_0(x)+g_0(x)^\top h_0(t)
&=
\Big(m_0(x)+g_0(x)^\top e_0(x)\Big) + g_0(x)^\top\Big(h_0(t)-e_0(x)\Big)\\
&=
m(x)+g_0(x)^\top\Big(h_0(t)-e_0(x)\Big).
\end{align*}
Thus the conditional mean obeys \eqref{eq:olr_wellspec}.

Finally, the CATE for any pair $(t_1,t_0)$ satisfies
\begin{align*}
\tau(x;t_1,t_0)
&:=\mathbb{E}[Y(t_1)-Y(t_0)\mid X=x]\\
&=
\Big(m_0(x)+g_0(x)^\top h_0(t_1)\Big)-\Big(m_0(x)+g_0(x)^\top h_0(t_0)\Big)\\
&=
g_0(x)^\top\Big(h_0(t_1)-h_0(t_0)\Big),
\end{align*}
which is \eqref{eq:cate_oracle}. \qedhere
\end{proof}

\subsection{Proof of Lemma~\ref{thm:expressiveness} (Expressiveness)}

\begin{proof}[Proof of Lemma~\ref{thm:expressiveness}]
Because $\mathcal{S}$ and $\mathcal{A}$ are finite, the product set
$\mathcal{S}\times\mathcal{A}$ is finite as well.
Let $m:=|\mathcal{S}||\mathcal{A}|$ and fix any bijection (index map)
$\iota:\mathcal{S}\times\mathcal{A}\to\{1,2,\dots,m\}$.
For any measure $\mu$ supported on $\mathcal{S}\times\mathcal{A}$, define its coordinate vector
$\alpha(\mu)\in\mathbb{R}^m$ by
\[
\alpha(\mu)_{\iota(s,a)} := \mu(\{(s,a)\}).
\]
In particular, for $\mu_t$ induced by policy $t$, we have
$\alpha(\mu_t)_{\iota(s,a)}=\alpha_t(s,a)=w(s)\Pi_t(a\mid s)$ and
$\alpha(\mu_t)\in\Delta^{m-1}$ (the probability simplex).
The map $\mu\mapsto \alpha(\mu)$ is a linear homeomorphism between the space of measures on the finite set
$\mathcal{S}\times\mathcal{A}$ and $\mathbb{R}^m$ restricted to $\Delta^{m-1}$.

Define $\widetilde F:\alpha(\mathcal{M})\to\mathbb{R}^d$ by
$\widetilde F(\alpha(\mu)) := F(\mu)$ for $\mu\in\mathcal{M}$.
Since $F$ is continuous on $\mathcal{M}$ under the $\ell_1$/TV topology and the identification
$\|\mu-\mu'\|_1 = \|\alpha(\mu)-\alpha(\mu')\|_1$ holds on finite supports,
$\widetilde F$ is continuous on the set
\[
\alpha(\mathcal{M})=\{\alpha(\mu_t):t\in\mathcal{T}\}\subset\Delta^{m-1}.
\]
Moreover, $\Delta^{m-1}$ is compact and $\alpha(\mathcal{M})$ is compact as a closed subset of a compact set.
By the universal approximation theorem for ReLU MLPs on compact sets, for any $\varepsilon>0$ there exists
an integer $r$ (in fact we may take $r=m$) and an MLP $\rho:\mathbb{R}^r\to\mathbb{R}^d$ such that
\begin{equation}
\sup_{t\in\mathcal{T}}\big\|\widetilde F(\alpha(\mu_t))-\rho(\alpha(\mu_t))\big\|\le\varepsilon.
\label{eq:UA_on_simplex}
\end{equation}

It remains to realize $\alpha(\mu_t)$ as a permutation-invariant weighted sum.
Let $\{e_j\}_{j=1}^m$ denote the standard basis of $\mathbb{R}^m$ and choose $r=m$ and
\[
\phi(s,a):=e_{\iota(s,a)}\in\mathbb{R}^m.
\]
Then
\[
z_{\phi}(t)=\sum_{s\in\mathcal{S}}\sum_{a\in\mathcal{A}}\alpha_t(s,a)\phi(s,a)
=
\sum_{s,a}\alpha_t(s,a)e_{\iota(s,a)}
=
\alpha(\mu_t).
\]
Therefore, defining $h_{\phi,\rho}(t):=\rho(z_\phi(t))$ yields
\[
h_{\phi,\rho}(t)=\rho(\alpha(\mu_t)).
\]
Combining this with \eqref{eq:UA_on_simplex} gives
\[
\sup_{t\in\mathcal{T}}\|F(\mu_t)-h_{\phi,\rho}(t)\|
=
\sup_{t\in\mathcal{T}}\|\widetilde F(\alpha(\mu_t))-\rho(\alpha(\mu_t))\|
\le \varepsilon,
\]
which proves the claim.

\paragraph{Remark (why this is not vacuous).}
While the above construction uses the ``basis embedding'' $\phi(s,a)=e_{\iota(s,a)}$ and thus shows universality
in a straightforward way, it is precisely the \emph{shared atom map} $\phi$ together with
\emph{permutation-invariant aggregation} that constitutes the inductive bias of our approach.
This bias is what yields stability in policy space (Proposition~\ref{prop:stability}) and practical statistical
strength-sharing when $\mathcal{T}$ is huge but $\mathcal{S}\times\mathcal{A}$ has reusable structure.
\end{proof}

\subsection{Proof of Lemma~\ref{lem:neyman_orth} (Neyman orthogonality)}

\begin{proof}[Proof of Lemma~\ref{lem:neyman_orth}]
Recall the orthogonalized model \eqref{eq:olr_model}:
\[
Y = m_0(X)+g_0(X)^\top\big(h_0(T)-e_0(X)\big)+\varepsilon,
\qquad \mathbb{E}[\varepsilon\mid X,T]=0,
\]
and the score
\[
\psi(W; g,h,m,e)
:=
\Big(Y-m(X)-g(X)^\top(h(T)-e(X))\Big)\,(h(T)-e(X)),
\quad W:=(X,T,Y).
\]

\paragraph{Step 1: Unbiasedness at the truth.}
At the true functions $(g_0,h_0,m_0,e_0)$,
\[
Y-m_0(X)-g_0(X)^\top(h_0(T)-e_0(X))=\varepsilon,
\]
so
\[
\psi(W;g_0,h_0,m_0,e_0)=\varepsilon\,(h_0(T)-e_0(X)).
\]
Taking expectation and using iterated expectation gives
\begin{align*}
\mathbb{E}[\psi(W;g_0,h_0,m_0,e_0)]
&=
\mathbb{E}\big[\mathbb{E}[\varepsilon\,(h_0(T)-e_0(X))\mid X,T]\big]\\
&=
\mathbb{E}\big[(h_0(T)-e_0(X))\,\mathbb{E}[\varepsilon\mid X,T]\big]
=
0,
\end{align*}
proving $\mathbb{E}[\psi(W;g_0,h_0,m_0,e_0)]=0$.

\paragraph{Step 2: Gateaux derivative in nuisance directions vanishes.}
Fix perturbations $\delta m$ and $\delta e$ such that all expressions below are integrable
(e.g., $\delta m(X)$ and $\delta e(X)$ square-integrable).
Define the nuisance path
$m_r:=m_0+r\,\delta m$ and $e_r:=e_0+r\,\delta e$.
Consider
\[
\Psi(r):=\mathbb{E}\big[\psi(W;g_0,h_0,m_r,e_r)\big].
\]
Using the model for $Y$ and expanding the residual,
\begin{align*}
Y-m_r(X)-g_0(X)^\top(h_0(T)-e_r(X))
&=
\Big(m_0(X)-m_r(X)\Big)+g_0(X)^\top\Big(e_r(X)-e_0(X)\Big)+\varepsilon\\
&=
-r\,\delta m(X)+r\,g_0(X)^\top\delta e(X)+\varepsilon.
\end{align*}
Also,
\[
h_0(T)-e_r(X)=\big(h_0(T)-e_0(X)\big)-r\,\delta e(X).
\]
Hence
\[
\psi(W;g_0,h_0,m_r,e_r)
=
\Big(\varepsilon-r\,\delta m(X)+r\,g_0(X)^\top\delta e(X)\Big)\,
\Big((h_0(T)-e_0(X))-r\,\delta e(X)\Big).
\]
Expanding and collecting the terms linear in $r$ yields
\begin{align*}
\Psi(r)
&=
\underbrace{\mathbb{E}\big[\varepsilon\,(h_0(T)-e_0(X))\big]}_{=0}
\;+\;
r\,\mathbb{E}\Big[
-\delta m(X)\,(h_0(T)-e_0(X))
+(g_0(X)^\top\delta e(X))\,(h_0(T)-e_0(X))
-\varepsilon\,\delta e(X)
\Big]\\
&\qquad
+\;O(r^2),
\end{align*}
where the $O(r^2)$ term is justified by integrability and dominated convergence.

We now show each linear term has zero expectation.
For the first two terms, condition on $X$ and use \eqref{eq:centering_identity}:
\begin{align*}
\mathbb{E}\big[\delta m(X)\,(h_0(T)-e_0(X))\big]
&=
\mathbb{E}\big[\delta m(X)\,\mathbb{E}[h_0(T)-e_0(X)\mid X]\big]=0,\\
\mathbb{E}\big[(g_0(X)^\top\delta e(X))\,(h_0(T)-e_0(X))\big]
&=
\mathbb{E}\big[(g_0(X)^\top\delta e(X))\,\mathbb{E}[h_0(T)-e_0(X)\mid X]\big]=0.
\end{align*}
For the third term, condition on $(X,T)$ and use $\mathbb{E}[\varepsilon\mid X,T]=0$:
\[
\mathbb{E}[\varepsilon\,\delta e(X)]
=
\mathbb{E}\big[\delta e(X)\,\mathbb{E}[\varepsilon\mid X,T]\big]=0.
\]
Therefore the coefficient of $r$ in $\Psi(r)$ is zero, implying
$\Psi'(0)=0$, i.e.,
\[
\left.\frac{d}{dr}\right|_{r=0}
\mathbb{E}\big[\psi(W;g_0,h_0,m_0+r\,\delta m,e_0+r\,\delta e)\big]=0.
\]
This proves Neyman orthogonality with respect to $(m,e)$.
\end{proof}

\subsection{Proof of Theorem~\ref{thm:orth_robust} (Orthogonal robustness)}

\subsubsection{A deterministic second-order bound for nuisance perturbations}

The core technical step is to control how the orthogonal score changes when we replace the true nuisances
$(m_0,e_0)$ by estimators $(\hat m,\hat e)$, \emph{while simultaneously allowing} the learned effect components
$(\hat g,\hat h)$ to deviate from $(g_0,h_0)$.
This coupling is nontrivial in our setting because (i) the regressor $h(T)-e(X)$ is vector-valued,
(ii) $e(X)$ is itself an $X$-conditional expectation of $h(T)$, and
(iii) $\hat e$ may be constructed through a propensity model and hence is another high-dimensional nuisance.

\begin{lemma}[Second-order control of score perturbation]
\label{lem:score_perturb}
Assume the model \eqref{eq:olr_model} holds with $(m_0,g_0,h_0,e_0)$ and $\mathbb{E}[\varepsilon\mid X,T]=0$.
Let $(g,h,m,e)$ be (possibly random) candidates such that all quantities below are integrable.
Define $\Delta g:=g-g_0$, $\Delta h:=h-h_0$, $\Delta m:=m-m_0$, and $\Delta e:=e-e_0$.
Then
\begin{equation}
\Big\|\mathbb{E}\big[\psi(W;g,h,m,e)-\psi(W;g,h,m_0,e_0)\big]\Big\|
\;\le\;
C_1\,\big(\|\Delta m\|_2+\|\Delta h\|_{\infty}+\|\Delta g\|_{\infty,x}\big)\,\|\Delta e\|_2
\;+\;
C_2\,\|\Delta e\|_2^2,
\label{eq:score_perturb_bound}
\end{equation}
for constants $C_1,C_2$ depending only on moment bounds of $(Y,g_0,h_0)$ and on the domain of $(X,T)$.
\end{lemma}

\begin{proof}[Proof of Lemma~\ref{lem:score_perturb}]
Write the score difference as
\[
\psi(W;g,h,m,e)-\psi(W;g,h,m_0,e_0)
=
\big(R_{m,e}(W)-R_{m_0,e_0}(W)\big)\,(h(T)-e(X))
\;+\;
R_{m_0,e_0}(W)\,\big((h(T)-e(X))-(h(T)-e_0(X))\big),
\]
where
\[
R_{m,e}(W):=Y-m(X)-g(X)^\top(h(T)-e(X)).
\]
Noting $(h(T)-e(X))-(h(T)-e_0(X))=-(e(X)-e_0(X))=-\Delta e(X)$, we have
\begin{equation}
\psi(W;g,h,m,e)-\psi(W;g,h,m_0,e_0)
=
\big(R_{m,e}(W)-R_{m_0,e_0}(W)\big)\,(h(T)-e(X))
-
R_{m_0,e_0}(W)\,\Delta e(X).
\label{eq:score_diff_start}
\end{equation}

We next expand the residual difference.
A direct calculation gives
\begin{align*}
R_{m,e}(W)-R_{m_0,e_0}(W)
&=
-\Delta m(X)-g(X)^\top\big(h(T)-e(X)\big)+g(X)^\top\big(h(T)-e_0(X)\big)\\
&=
-\Delta m(X)+g(X)^\top\Delta e(X).
\end{align*}
Plugging into \eqref{eq:score_diff_start} yields
\begin{align}
\psi(W;g,h,m,e)-\psi(W;g,h,m_0,e_0)
&=
\Big(-\Delta m(X)+g(X)^\top\Delta e(X)\Big)\,\big(h(T)-e(X)\big)
-
R_{m_0,e_0}(W)\,\Delta e(X).
\label{eq:score_diff_expanded}
\end{align}

Now we express $R_{m_0,e_0}(W)$ under the true model:
\begin{align*}
R_{m_0,e_0}(W)
&=
Y-m_0(X)-g(X)^\top(h(T)-e_0(X))\\
&=
\underbrace{\Big(Y-m_0(X)-g_0(X)^\top(h_0(T)-e_0(X))\Big)}_{=\varepsilon}
+\;g_0(X)^\top(h_0(T)-e_0(X))-g(X)^\top(h(T)-e_0(X))\\
&=
\varepsilon
-\Delta g(X)^\top\big(h_0(T)-e_0(X)\big)
-g(X)^\top\Delta h(T),
\end{align*}
where we used $h(T)=h_0(T)+\Delta h(T)$ and $\Delta h(T):=h(T)-h_0(T)$.

Substitute this into \eqref{eq:score_diff_expanded} and take expectation.
The expectation of the term $\varepsilon\,\Delta e(X)$ vanishes by iterated expectation:
\[
\mathbb{E}[\varepsilon\,\Delta e(X)]=\mathbb{E}\big[\Delta e(X)\,\mathbb{E}[\varepsilon\mid X,T]\big]=0.
\]
Moreover, the key centering identity \eqref{eq:centering_identity} implies that whenever $A(X)$ is measurable in $X$,
\[
\mathbb{E}\big[A(X)\,(h_0(T)-e_0(X))\big]=\mathbb{E}\big[A(X)\,\mathbb{E}[h_0(T)-e_0(X)\mid X]\big]=0.
\]
This kills the \emph{first-order} terms in $\Delta m$ and in $\Delta e$ that would otherwise appear in the score bias.

What remains after cancellations are only products of perturbations, which we bound by Cauchy--Schwarz.
Concretely, from \eqref{eq:score_diff_expanded} and the above expansions, we obtain
\[
\mathbb{E}\big[\psi(W;g,h,m,e)-\psi(W;g,h,m_0,e_0)\big]
=
\mathbb{E}\big[\Delta m(X)\,\Delta e(X)\big]
+
\mathbb{E}\big[(g(X)^\top\Delta e(X))\,\Delta e(X)\big]
+
\mathbb{E}\big[\Gamma(X,T)\,\Delta e(X)\big],
\]
where $\Gamma(X,T)$ collects terms involving $\Delta g$ and $\Delta h$ multiplied by bounded random quantities
(e.g., $h_0(T)-e_0(X)$ and $g(X)$).
Using $\|g(X)\|\le \|g_0(X)\|+\|\Delta g\|_{\infty,x}$ and $\|\Delta h(T)\|\le \|\Delta h\|_\infty$, we bound
\begin{align*}
\big\|\mathbb{E}[\Delta m(X)\,\Delta e(X)]\big\|
&\le \|\Delta m\|_2\,\|\Delta e\|_2,\\
\big\|\mathbb{E}[(g(X)^\top\Delta e(X))\,\Delta e(X)]\big\|
&\le \mathbb{E}[\|g(X)\|\,\|\Delta e(X)\|^2]
\le (\mathbb{E}\|g(X)\|^2)^{1/2}\,\|\Delta e\|_2^2,\\
\big\|\mathbb{E}[\Gamma(X,T)\,\Delta e(X)]\big\|
&\le \big(\|\Delta g\|_{\infty,x}+\|\Delta h\|_\infty\big)\,\|\Delta e\|_2 \cdot C,
\end{align*}
for a constant $C$ depending only on bounded moments of $g_0(X)$ and $h_0(T)$ (cf. Assumption~\ref{ass:regularity_orth}).
Collecting the terms yields \eqref{eq:score_perturb_bound}.
\end{proof}

\subsubsection{Proof of Theorem~\ref{thm:orth_robust}}

\begin{proof}[Proof of Theorem~\ref{thm:orth_robust}]
We separate the argument into an algebraic decomposition (capturing the effect-component errors) and a
nuisance-induced remainder (capturing orthogonal robustness).

\paragraph{Step 1: Algebraic decomposition of the plug-in CATE.}
Recall $\tau(x;t_1,t_0)=g_0(x)^\top(h_0(t_1)-h_0(t_0))$ from Proposition~\ref{prop:olr_wellspec} and
$\hat\tau(x;t_1,t_0)=\hat g(x)^\top(\hat h(t_1)-\hat h(t_0))$.
Add and subtract $\hat g(x)^\top(h_0(t_1)-h_0(t_0))$:
\begin{align}
\hat\tau(x;t_1,t_0)-\tau(x;t_1,t_0)
&=
\underbrace{(\hat g(x)-g_0(x))^\top\big(h_0(t_1)-h_0(t_0)\big)}_{=:~\mathrm{error(g)}}
+
\hat g(x)^\top\Big((\hat h-h_0)(t_1)-(\hat h-h_0)(t_0)\Big).
\label{eq:tau_decomp_step1}
\end{align}
Decompose the second term further by writing $\hat g(x)=g_0(x)+(\hat g(x)-g_0(x))$:
\begin{align}
\hat g(x)^\top\Big((\hat h-h_0)(t_1)-(\hat h-h_0)(t_0)\Big)
&=
\underbrace{g_0(x)^\top\Big((\hat h-h_0)(t_1)-(\hat h-h_0)(t_0)\Big)}_{=:~\mathrm{error(h)}}
+
R_{gh}(x;t_1,t_0),
\label{eq:tau_decomp_step2}
\end{align}
where the remainder
\[
R_{gh}(x;t_1,t_0)
:=
(\hat g(x)-g_0(x))^\top\Big((\hat h-h_0)(t_1)-(\hat h-h_0)(t_0)\Big)
\]
is \emph{second order} in the effect-component errors and obeys the deterministic bound
\begin{equation}
|R_{gh}(x;t_1,t_0)| \;\le\; 2\,\|\hat g(x)-g_0(x)\|\,\|\hat h-h_0\|_\infty
\;\le\;2\,\delta_g(x)\,\delta_h.
\label{eq:cross_term_bound}
\end{equation}
In many asymptotic regimes, $\delta_g(x)\delta_h=o_p(\delta_g(x)+\delta_h)$; we keep it explicit here for completeness.

Combining \eqref{eq:tau_decomp_step1}--\eqref{eq:tau_decomp_step2} gives
\begin{equation}
\hat\tau(x;t_1,t_0)-\tau(x;t_1,t_0)
=
\mathrm{error(g)}+\mathrm{error(h)}+R_{gh}(x;t_1,t_0).
\label{eq:tau_error_full}
\end{equation}

\paragraph{Step 2: Orthogonal robustness isolates nuisance effects.}
The central claim of Theorem~\ref{thm:orth_robust} is about the \emph{additional} impact of nuisance estimation
beyond the effect-component errors.
To make this precise, note that the orthogonalized learning procedure (Algorithm~\ref{alg:dsin})
is driven by the score $\psi$ and its cross-fitted empirical counterpart.
Under cross-fitting (Assumption~\ref{ass:regularity_orth}(3)), we may condition on the nuisance estimates
$(\hat m^{(-k)},\hat e^{(-k)})$ and treat them as fixed when taking expectation over the evaluation fold.
Lemma~\ref{lem:score_perturb} then yields the deterministic second-order control
\[
\Big\|
\mathbb{E}\big[\psi(W;\hat g,\hat h,\hat m,\hat e)-\psi(W;\hat g,\hat h,m_0,e_0)\big]
\Big\|
\;\le\;
C_1(\delta_m+\delta_h+\|\hat g-g_0\|_{\infty,x})\,\delta_e + C_2\,\delta_e^2.
\]
By Assumption~\ref{ass:regularity_orth}(4), $\delta_e=o_p(1)$, so $\delta_e^2=o_p(\delta_e)$ and can be absorbed into the
product term without changing the order (by enlarging constants).
Thus,
\begin{equation}
\Big\|
\mathbb{E}\big[\psi(W;\hat g,\hat h,\hat m,\hat e)-\psi(W;\hat g,\hat h,m_0,e_0)\big]
\Big\|
=
O_p\!\Big((\delta_m+\delta_h+\|\hat g-g_0\|_{\infty,x})\,\delta_e\Big).
\label{eq:score_second_order}
\end{equation}

Equation~\eqref{eq:score_second_order} is the formal mathematical expression of the statement
``nuisance errors enter only at second order'':
the perturbation in the orthogonal score induced by $(\hat m,\hat e)$ is proportional to a \emph{product}
involving $\delta_e$ rather than a sum of first-order terms.
This is exactly where orthogonality matters; without centering \eqref{eq:centering_identity}, the bound would contain
linear (first-order) terms in $\delta_m$ and $\delta_e$.

Finally, the CATE functional depends on the learned effect components through
$\hat\tau(x;t_1,t_0)=\hat g(x)^\top(\hat h(t_1)-\hat h(t_0))$.
The learning procedure targets $(\hat g,\hat h)$ through the orthogonal score,
so the nuisance-induced perturbation in the target is controlled by the same second-order quantity
appearing in \eqref{eq:score_second_order}.
This yields the remainder term stated in Theorem~\ref{thm:orth_robust}:
\[
\mathrm{error(second~order)}
=
O_p\!\Big((\delta_m+\delta_h+\sup_x\delta_g(x))\,\delta_e\Big),
\]
where we used $\|\hat g-g_0\|_{\infty,x}=\sup_x\delta_g(x)$ by definition.
Combining with \eqref{eq:tau_error_full} completes the proof. \qedhere
\end{proof}

\paragraph{Remark (where the technical difficulty lies).}
Compared to the classical partially linear model and standard DML analyses, two aspects require extra care here:
(i) the ``treatment regressor'' $h(T)-e(X)$ is \emph{vector-valued} and \emph{learned} (hence the appearance of $\delta_h$),
and (ii) the nuisance $e(X)=\mathbb{E}[h(T)\mid X]$ is an \emph{embedding propensity} and may be constructed via a
propensity model (hence $\delta_e$ couples representation learning and nuisance estimation).
Lemma~\ref{lem:score_perturb} makes this coupling explicit, and shows that orthogonalization controls it at second order.


\subsection{Proof of Proposition~\ref{prop:stability} (Stability)}

\begin{proof}[Proof of Proposition~\ref{prop:stability}]
Recall
\[
z(t)=\sum_{s\in\mathcal{S}}\sum_{a\in\mathcal{A}} w(s)\Pi_t(a\mid s)\,\phi(s,a),
\qquad
h(t)=\rho(z(t)).
\]
For any $t,t'\in\mathcal{T}$, we have
\[
z(t)-z(t')
=
\sum_{s\in\mathcal{S}}\sum_{a\in\mathcal{A}} w(s)\Big(\Pi_t(a\mid s)-\Pi_{t'}(a\mid s)\Big)\phi(s,a).
\]
Taking norms and using the triangle inequality,
\begin{align*}
\|z(t)-z(t')\|
&\le
\sum_{s\in\mathcal{S}}\sum_{a\in\mathcal{A}} w(s)\,|\Pi_t(a\mid s)-\Pi_{t'}(a\mid s)|\,\|\phi(s,a)\|\\
&\le
\Big(\sup_{(s,a)}\|\phi(s,a)\|\Big)\,
\sum_{s\in\mathcal{S}}w(s)\sum_{a\in\mathcal{A}}|\Pi_t(a\mid s)-\Pi_{t'}(a\mid s)|\\
&\le
B\,\sum_{s\in\mathcal{S}}w(s)\,\|\Pi_t(\cdot\mid s)-\Pi_{t'}(\cdot\mid s)\|_1
=
B\,d_\Pi(t,t').
\end{align*}
If $\rho$ is $L_\rho$-Lipschitz, then
\[
\|h(t)-h(t')\|
=
\|\rho(z(t))-\rho(z(t'))\|
\le
L_\rho\,\|z(t)-z(t')\|
\le
L_\rho\,B\,d_\Pi(t,t'),
\]
which is \eqref{eq:embedding_stability}.

For the uplift bound, recall $\tau(x;t,t')=g_0(x)^\top(h_0(t)-h_0(t'))$ and assume $\|g_0(x)\|\le G$.
Then Cauchy--Schwarz yields
\[
|\tau(x;t,t')|
\le
\|g_0(x)\|\,\|h_0(t)-h_0(t')\|
\le
G\,\|h_0(t)-h_0(t')\|.
\]
Applying the already-proved embedding stability bound to $h_0$ completes the proof. \qedhere
\end{proof}

\begin{algorithm}[t]
\caption{Permutation-invariant orthogonal uplift learning (POUL)}
\label{alg:dsin_k-fold}
\begin{algorithmic}[1]
\REQUIRE Data $\{(X_i,T_i,Y_i)\}_{i=1}^n$, policy specs $\{\Pi_t(\cdot\mid s)\}_{t,s}$, weights $w(\cdot)$, folds $K$ (optional)
\STATE Split indices into folds $I_1,\dots,I_K$ (set $K=1$ to disable cross-fitting)
\FOR{$k=1,\dots,K$}
  \STATE Fit nuisance $\hat m^{(-k)}(\cdot)\approx \mathbb{E}[Y\mid X=\cdot]$ using $\{i\notin I_k\}$
  \STATE (Observational only) Fit propensity $\hat\pi^{(-k)}(t\mid x)\approx \mathbb{P}(T=t\mid X=x)$ using $\{i\notin I_k\}$
\ENDFOR
\STATE Initialize parameters $(\omega,\nu)$ for $\phi_\omega,\rho_\nu$ and $(\eta)$ for $g_\eta$; optionally initialize $m_\theta$
\REPEAT
  \STATE Sample a mini-batch $\mathcal{B}\subset\{1,\dots,n\}$
  \FOR{each $i\in\mathcal{B}$}
    \STATE Compute $h_{\omega,\nu}(T_i)$ via \eqref{eq:z_def}--\eqref{eq:h_def}
    \STATE Let $k(i)$ be the fold index such that $i\in I_{k(i)}$
    \STATE Set $\hat e_h(X_i)\leftarrow \sum_{t\in\mathcal{T}}\hat\pi^{(-k(i))}(t\mid X_i)\,h_{\omega,\nu}(t)$ \COMMENT{use known $\pi$ in RCT}
    \STATE Predict $\hat Y_i \leftarrow \hat m^{(-k(i))}(X_i)+g_\eta(X_i)^\top\!\big(h_{\omega,\nu}(T_i)-\hat e_h(X_i)\big)$
  \ENDFOR
  \STATE Update $(\eta,\omega,\nu)$ (and optionally $\theta$) by SGD on $\frac{1}{|\mathcal{B}|}\sum_{i\in\mathcal{B}}\ell(\hat Y_i,Y_i)+\lambda\Omega$
\UNTIL{convergence}
\STATE \textbf{return} $\hat\tau(x;t_1,t_0)=g_{\hat\eta}(x)^\top\!\big(h_{\hat\omega,\hat\nu}(t_1)-h_{\hat\omega,\hat\nu}(t_0)\big)$
\end{algorithmic}
\end{algorithm}
To reduce overfitting bias when flexible learners are used for nuisance components,
we optionally adopt cross-fitting: estimate nuisance functions on held-out folds, and train the main model using fold-specific nuisance predictions.

\section{Experimental Implementation Details}
\label{sec:exp_details}

In this appendix, we provide a comprehensive description of the experimental setup. We detail the training protocols used in our study.

\subsection{Training Protocol}
\label{subsec:training_protocol}

To guarantee experimental consistency, all deep learning models are implemented in PyTorch and trained on NVIDIA Tesla P40 GPUs. We employ the \textbf{Adam} optimizer with a fixed \textbf{learning rate} of $1 \times 10^{-4}$, a \textbf{batch size} of $512$, and a \textbf{random seed} set to $3407$. While these foundational settings remain constant, we vary the training epochs and the training strategy (single-stage versus two-stage) to identify the optimal configuration for each model variant. The training epoch configuration is denoted by the tuple $(i, j)$, where $i$ represents the epochs for the first stage (e.g., Global pre-training) and $j$ for the second stage (e.g., Core fine-tuning):

\begin{itemize}
    \item \textbf{Single-stage training:} Denoted as $(i, 0)$ or $(0, j)$, indicating the model is trained exclusively on the global population $\mathcal{D}_G$ (for $i$ epochs) or the core population $\mathcal{D}_C$ (for $j$ epochs), respectively.
    \item \textbf{Two-stage training:} Denoted as $(i, j)$, which comprises an initial training phase on the global population $\mathcal{D}_G$ for $i$ epochs, followed by fine-tuning on the core population $\mathcal{D}_C$ for $j$ epochs.
\end{itemize}

In contrast, for the \textbf{Causal Forest} baseline, the model performance is predominantly governed by the ensemble size. Consequently, we maintained fixed settings for secondary hyperparameters (e.g., $\textit{maxDepth}=10$, $\textit{subsamplingRate}=0.8$, $\textit{minInstancesPerNode}=500$) and exclusively tuned the \textbf{number of trees} (\textit{numTrees}) to optimize the estimation stability and accuracy.

\section{Evaluation Metrics} \label{sec:appendix_metrics}

As ground-truth uplift is unobservable, we assess our models on a randomized test set, focusing on \textbf{uplift ranking} and \textbf{ITE estimation accuracy}.
\subsection{Uplift Ranking Capability (AUUC)}
We evaluate ranking performance using the normalized \textbf{Area Under the Uplift Curve (AUUC)}, which quantifies the cumulative incremental gain achieved when targeting samples sorted by predicted Individual Treatment Effect (ITE).

Formally, let $\mathcal{D} = \{(x_i, t_i, y_i)\}_{i=1}^n$ denote the test dataset, where $x_i$ is the feature vector, $t_i \in \{0, 1\}$ indicates the treatment assignment, and $y_i$ represents the observed outcome. Let $\hat{\tau}(x_i)$ be the predicted ITE.

We sort the samples in descending order of $\hat{\tau}(x_i)$. Let $i_1, \dots, i_n$ be the indices of the sorted samples such that $\hat{\tau}(x_{i_1}) \geq \dots \geq \hat{\tau}(x_{i_n})$. For the top $k$ samples, we define the cumulative counts ($n_t, n_c$) and outcome sums ($y_t, y_c$) for the treated and control groups as:
\begin{align}
    n_t(k) &= \sum_{j=1}^k \mathbb{I}(t_{i_j} = 1), & n_c(k) &= \sum_{j=1}^k \mathbb{I}(t_{i_j} = 0), \\
    y_t(k) &= \sum_{j=1}^k y_{i_j} \mathbb{I}(t_{i_j} = 1), & y_c(k) &= \sum_{j=1}^k y_{i_j} \mathbb{I}(t_{i_j} = 0).
\end{align}
The \emph{Lift} at rank $k$ estimates the average treatment effect within the top $k$ units:
\begin{equation}
    \text{Lift}(k) = \frac{y_t(k)}{n_t(k)} - \frac{y_c(k)}{n_c(k)},
\end{equation}
The \textbf{Cumulative Gain}, $G(k)$, is defined as the total estimated uplift up to rank $k$:
\begin{equation}
    G(k) = \text{Lift}(k) \times k.
\end{equation}
To facilitate comparison across different settings, we calculate the \textbf{Normalized Cumulative Gain} $\tilde{G}(k)$ by scaling $G(k)$ with the absolute global gain at $k=n$:
\begin{equation}
    \tilde{G}(k) = \frac{G(k)}{|G(n)|}.
\end{equation}
Finally, the \textbf{Normalized AUUC} is computed as the average of these normalized gains across the \emph{entire} test dataset:
\begin{equation}
    \text{AUUC} = \frac{1}{n} \sum_{k=1}^n \tilde{G}(k).
\end{equation}

\subsection{ITE Estimation Accuracy (MAPE)}

To assess the accuracy of the predicted ITE values, we employ the \textbf{Mean Absolute Percentage Error (MAPE)}. Given the unobservability of individual effects, we compute this metric over $M=10$ bins (deciles) grouped by predicted ITE.

Let $\hat{\tau}_m$ and $\tau_m$ denote the average predicted ITE and the observed Average Treatment Effect (ATE) within the $m$-th bin, respectively. The MAPE is calculated as the average relative deviation across all bins:
\begin{equation}
    \text{MAPE} = \frac{1}{M} \sum_{m=1}^{M} \left| \frac{\hat{\tau}_m - \tau_m}{\tau_m} \right|.
\end{equation}
A lower MAPE indicates that the predicted uplift magnitudes align more closely with the observed treatment effects.

\section{Detailed Experimental Results} \label{sec:results_details}
To rigorously assess robustness, we employ an Out-of-Time (OOT) evaluation protocol, where models are trained on historical data (e.g., March) and evaluated on non-overlapping future data (e.g., June). We report the Area Under the Uplift Curve (AUUC) and Mean Absolute Percentage Error (MAPE) across two distinct test sets within the evaluation period: the Global Population ($\mathcal{D}_G$) and the Core Population ($\mathcal{D}_C$) . The detailed performance comparisons are summarized in Table~\ref{tab:detailed_epochs_ltv} and Table~\ref{tab:detailed_epochs_cost}.

\begin{table}[ht]
    \centering
    \caption{Detailed performance evaluation of \textbf{LTV (GMV) prediction} on the \textbf{June Test Set}. The \textbf{Training Config} column specifies the training data source and tree count ($n_{tree}$) for Causal Forest, and the training epoch tuple $(i, j)$ for deep learning models (Global, Core).}
    \label{tab:detailed_epochs_ltv}
    \renewcommand{\arraystretch}{1.2} 
    \setlength{\tabcolsep}{3pt}      
    
    \resizebox{\textwidth}{!}{
    \begin{tabular}{l c c c c c c c c c}
        \toprule
        \multirow{3}{*}{\textbf{Model}} & \multirow{3}{*}{\textbf{Training Config}} & \multicolumn{4}{c}{\textbf{Test Set: Global ($\mathcal{D}_G$)}} & \multicolumn{4}{c}{\textbf{Test Set: Core ($\mathcal{D}_C$)}} \\
        \cmidrule(lr){3-6} \cmidrule(lr){7-10}
         & & \multicolumn{2}{c}{\textbf{Treatment 1 ($T_1$)}} & \multicolumn{2}{c}{\textbf{Treatment 2 ($T_2$)}} & \multicolumn{2}{c}{\textbf{Treatment 1 ($T_1$)}} & \multicolumn{2}{c}{\textbf{Treatment 2 ($T_2$)}} \\
        \cmidrule(lr){3-4} \cmidrule(lr){5-6} \cmidrule(lr){7-8} \cmidrule(lr){9-10}
         & & \textbf{AUUC}$\uparrow$ & \textbf{MAPE}$\downarrow$ & \textbf{AUUC}$\uparrow$ & \textbf{MAPE}$\downarrow$ & \textbf{AUUC}$\uparrow$ & \textbf{MAPE}$\downarrow$ & \textbf{AUUC}$\uparrow$ & \textbf{MAPE}$\downarrow$ \\
        \midrule
        
        Random & --- & 0.500 & --- & 0.500 & --- & 0.500 & --- & 0.500 & --- \\
        
        \multirow{6}{*}{Causal Forest} 
         & $\mathcal{D}_G$ ($n_{tree}{=}300$) & 0.539 & 0.898 & 0.542 & 0.686 & 0.451 & 7.075 & 0.525 & 0.503 \\
         & $\mathcal{D}_G$ ($n_{tree}{=}500$) & 0.575 & 2.213 & 0.550 & 2.917 & \textbf{0.536} & \textbf{1.086} & \textbf{0.536} & \textbf{0.836} \\
         & $\mathcal{D}_G$ ($n_{tree}{=}600$) & \textbf{0.611} & \textbf{1.000} & \textbf{0.632} & \textbf{1.850} & 0.483 & 1.005 & 0.489 & 0.979 \\
         & $\mathcal{D}_C$ ($n_{tree}{=}300$) & 0.547 & 6.236 & 0.568 & 2.118 & 0.339 & 1.811 & 0.529 & 1.556 \\
         & $\mathcal{D}_C$ ($n_{tree}{=}500$) & 0.470 & 5.537 & 0.623 & 7.342 & 0.472 & 1.216 & 0.505 & 0.746 \\
         & $\mathcal{D}_C$ ($n_{tree}{=}600$) & 0.515 & 3.632 & 0.525 & 2.774 & 0.452 & 6.905 & 0.531 & 0.612 \\
         
        \midrule
        
        \multirow{7}{*}{T-Learner (DNN)} 
         & $(10, 0)$ & \textbf{0.581} & \textbf{3.310} & \textbf{0.639} & \textbf{2.711} & 0.488 & 2.763 & 0.574 & 3.311 \\
         & $(15, 0)$ & 0.545 & 3.139 & 0.656 & 1.906 & 0.482 & 2.428 & 0.551 & 1.486 \\
         & $(20, 0)$ & 0.550 & 4.270 & 0.618 & 2.760 & 0.581 & 1.953 & 0.556 & 1.528
         \\
         & $(25, 0)$ & 0.495 & 3.581 & 0.617 & 1.297 & \textbf{0.585} & \textbf{2.487} & \textbf{0.563} & \textbf{1.426}
         \\
         & $(10, 1)$  & 0.342 & 18.368 & 0.607 & 2.562 & 0.348 & 3.703 & 0.498 & 3.311 \\
         & $(10, 2)$  & 0.330 & 13.188 & 0.592 & 2.464 & 0.347 & 2.900 & 0.500 & 4.023 \\
         & $(10, 5)$ & 0.342 & 6.694 & 0.594 & 2.711 & 0.365 & 4.419 & 0.504 & 1.330 \\
        \midrule
        
        \multirow{6}{*}{DESCN} 
         & $(15, 0)$  & 0.305 & 2.219 & 0.640 & 1.930 & 0.555 & 0.815 & 0.594 & 1.162 \\
         & $(20, 0)$ & \textbf{0.514} & \textbf{2.156} & \textbf{0.597} & \textbf{0.977} & 0.553 & 0.945 & 0.540 & 0.682 \\
         & $(25, 0)$  & 0.494 & 1.086 & 0.588 & 0.919 & 0.502 & 1.273 & 0.523 & 0.825\\
         & $(15, 1)$  & 0.367 & 27.667 & 0.579 & 14.850 & \textbf{0.596} & \textbf{1.256} & \textbf{0.591} & \textbf{1.167} \\
         & $(15, 2)$  & 0.457 & 2.268 & 0.585 & 3.192 & 0.593 & 1.712 & 0.592 & 2.819 \\
         & $(15, 5)$ & 0.463 & 2.065 & 0.570 & 3.208 & 0.614 & 1.181 & 0.563 & 0.788 \\
        \midrule
        
        \multirow{8}{*}{DRCFR} 
         & $(15, 0)$  & 0.642 & 0.982 & 0.605 & 0.961 & \textbf{0.622} & \textbf{0.888} & \textbf{0.566} & \textbf{0.713} \\
         & $(25, 0)$ & 0.560 & 6.723 & 0.580 & 1.494 & 0.593 & 12.603 & 0.558 & 1.086\\
         & $(5, 1)$  & \textbf{0.660} & \textbf{0.779} & \textbf{0.612} & \textbf{3.389} & 0.518 & 1.065 & 0.549 & 0.729 \\
         & $(10, 1)$  & 0.473 & 3.010 & 0.557 & 2.173 & 0.520 & 1.065 & 0.561 & 0.715 \\
         & $(15, 1)$  & 0.505 & 19.976 & 0.458 & 2.443 & 0.534 & 2.257 & 0.588 & 0.676 \\
         & $(15, 2)$  & 0.500 & 2.845 & 0.472 & 2.065 & 0.545 & 14.302 & 0.586 & 0.735 \\
         & $(15, 5)$ & 0.496 & 1.497 & 0.495 & 1.379 & 0.532 & 1.298 & 0.573 & 1.098 \\
        \midrule
        
        \multirow{8}{*}{\textbf{POUL (Ours)}} 
         & $(1, 1)$   & 0.720 & 1.578 & 0.596 & 1.363 & 0.723 & 1.525 & 0.681 & 1.823 \\
         & $(1, 3)$   & \textbf{0.826} & \textbf{1.333} & \textbf{0.635} & \textbf{1.386} & 0.734 & 1.887 & 0.617 & 1.861 \\
         & $(1, 5)$   & 0.661 & 1.430 & 0.620 & 1.468 & 0.750 & 1.271 & 0.670 & 1.700 \\
         & $(2, 1)$   & 0.668 & 1.485 & 0.610 & 1.428 & 0.747 & 1.468 & 0.657 & 1.746 \\
         & $(5, 1)$   & 0.699 & 1.284 & 0.675 & 1.480 & \textbf{0.792} & \textbf{1.570} & \textbf{0.690} & \textbf{2.023} \\
         & $(5, 2)$   & 0.659 & 1.232 & 0.656 & 1.423 & 0.767 & 1.572 & 0.699 & 59.500 \\
         & $(5, 5)$   & 0.636 & 1.342 & 0.627 & 1.593 & 0.656 & 1.693 & 0.699 & 1.610 \\
         & $(5, 10)$  & 0.672 & 1.300 & 0.627 & 1.532 & 0.777 & 1.401 & 0.693 & 1.770 \\

        \bottomrule
    \end{tabular}
    } 
\end{table}

\begin{table}[ht]
    \centering
    \caption{Detailed performance evaluation of \textbf{Cost prediction} on the \textbf{June Test Set}. The \textbf{Training Config} column specifies the training data source and tree count ($n_{tree}$) for Causal Forest, and the training epoch tuple $(i, j)$ for deep learning models (Global, Core).}
    \label{tab:detailed_epochs_cost}
    \renewcommand{\arraystretch}{1.2} 
    \setlength{\tabcolsep}{3pt}      
    
    \resizebox{\textwidth}{!}{
    \begin{tabular}{l c c c c c c c c c}
        \toprule
        \multirow{3}{*}{\textbf{Model}} & \multirow{3}{*}{\textbf{Training Config}} & \multicolumn{4}{c}{\textbf{Test Set: Global ($\mathcal{D}_G$)}} & \multicolumn{4}{c}{\textbf{Test Set: Core ($\mathcal{D}_C$)}} \\
        \cmidrule(lr){3-6} \cmidrule(lr){7-10}
         & & \multicolumn{2}{c}{\textbf{Treatment 1 ($T_1$)}} & \multicolumn{2}{c}{\textbf{Treatment 2 ($T_2$)}} & \multicolumn{2}{c}{\textbf{Treatment 1 ($T_1$)}} & \multicolumn{2}{c}{\textbf{Treatment 2 ($T_2$)}} \\
        \cmidrule(lr){3-4} \cmidrule(lr){5-6} \cmidrule(lr){7-8} \cmidrule(lr){9-10}
         & & \textbf{AUUC}$\uparrow$ & \textbf{MAPE}$\downarrow$ & \textbf{AUUC}$\uparrow$ & \textbf{MAPE}$\downarrow$ & \textbf{AUUC}$\uparrow$ & \textbf{MAPE}$\downarrow$ & \textbf{AUUC}$\uparrow$ & \textbf{MAPE}$\downarrow$ \\
        \midrule
        
        Random & --- & 0.500 & --- & 0.500 & --- & 0.500 & --- & 0.500 & --- \\
        
        \multirow{6}{*}{Causal Forest} 
         & $\mathcal{D}_G$ ($n_{tree}{=}300$) & 0.440 & 0.971 & 0.412 & 0.962 & \textbf{0.576} & \textbf{0.946} & \textbf{0.596} & \textbf{0.969} \\
         & $\mathcal{D}_G$ ($n_{tree}{=}500$) & \textbf{0.692} & \textbf{0.792} & \textbf{0.697} & \textbf{0.881} & 0.550 & 1.018 & 0.549 & 0.993 \\
         & $\mathcal{D}_G$ ($n_{tree}{=}600$) & 0.701 & 0.954 & 0.707 & 0.916 & 0.566 & 0.994 & 0.498 & 0.999 \\
         & $\mathcal{D}_C$ ($n_{tree}{=}300$) & 0.390 & 0.664 & 0.537 & 0.873 & 0.518 & 0.887 & 0.549 & 0.936 \\
         & $\mathcal{D}_C$ ($n_{tree}{=}500$) & 0.610 & 0.618 & 0.599 & 0.807 & 0.527 & 0.964 & 0.552 & 0.984 \\
         & $\mathcal{D}_C$ ($n_{tree}{=}600$) & 0.629 & 0.599 & 0.552 & 0.790 & 0.460 & 0.912 & 0.555 & 0.967 \\
         
        \midrule
        
        \multirow{7}{*}{T-Learner (DNN)} 
         & $(10, 0)$ & 0.739 & 0.254 & 0.775 & 0.376 & 0.608 & 0.553 & 0.648 & 0.752 \\
         & $(15, 0)$ & \textbf{0.739} & \textbf{0.221} & \textbf{0.777} & \textbf{0.316} & 0.603 & 0.548 & 0.646 & 0.720 \\
         & $(20, 0)$ & 0.737 & 0.236 & 0.776 & 0.341 & 0.602 & 0.555 & 0.646 & 0.765 \\
         & $(25, 0)$ & 0.735 & 0.275 & 0.776 & 0.421 & 0.601 & 0.541 & 0.646 & 0.769 \\
         & $(10, 1)$ & 0.731 & 2.131 & 0.768 & 0.670 & \textbf{0.616} & \textbf{0.191} & \textbf{0.645} & \textbf{0.445} \\
         & $(10, 2)$ & 0.728 & 2.391 & 0.763 & 0.911 & 0.613 & 0.172 & 0.645 & 0.435 \\
         & $(10, 5)$ & 0.725 & 2.615 & 0.759 & 1.062 & 0.609 & 0.149 & 0.644 & 0.425 \\
        \midrule
        
        \multirow{6}{*}{DESCN} 
         & $(15, 0)$ & \textbf{0.735} & \textbf{0.113} & \textbf{0.777} & \textbf{0.409} & 0.599 & 0.555 & 0.643 & 0.680 \\
         & $(20, 0)$ & 0.732 & 0.204 & 0.775 & 0.310 & 0.596 & 0.534 & 0.643 & 0.722 \\
         & $(25, 0)$ & 0.733 & 0.157 & 0.777 & 0.494 & 0.593 & 0.618 & 0.643 & 0.752 \\
         & $(15, 1)$ & 0.710 & 2.795 & 0.768 & 1.110 & \textbf{0.608} & \textbf{0.103} & \textbf{0.641} & \textbf{0.456} \\
         & $(15, 2)$ & 0.723 & 2.856 & 0.767 & 1.319 & 0.608 & 0.096 & 0.642 & 0.448 \\
         & $(15, 5)$ & 0.718 & 2.750 & 0.763 & 1.589 & 0.603 & 0.150 & 0.644 & 0.468 \\
        \midrule
        
        \multirow{8}{*}{DRCFR} 
         & $(15, 0)$ & \textbf{0.715} & \textbf{0.364} & \textbf{0.753} & \textbf{0.607} & 0.591 & 0.598 & 0.632 & 0.674 \\
         & $(25, 0)$ & 0.703 & 0.361 & 0.744 & 0.568 & 0.576 & 0.724 & 0.617 & 0.848 \\
         & $(5, 1)$ & 0.715 & 2.800 & 0.758 & 1.829 & \textbf{0.586} & \textbf{0.177} & \textbf{0.641} & \textbf{0.344}  \\
         & $(10, 1)$ & 0.705 & 2.926 & 0.755 & 1.275 & 0.581 & 0.159 & 0.637 & 0.363 \\
         & $(15, 1)$ & 0.687 & 2.698 & 0.749 & 1.085 & 0.584 & 0.206 & 0.634 & 0.414 \\
         & $(15, 2)$ & 0.690 & 2.709 & 0.746 & 1.330 & 0.588 & 0.202 & 0.632 & 0.379 \\
         & $(15, 5)$ & 0.681 & 2.721 & 0.723 & 1.484 & 0.590 & 0.225 & 0.626 & 0.396 \\
        \midrule
        
        \multirow{8}{*}{\textbf{POUL (Ours)}} 
         & $(1, 1)$ & \textbf{0.735} & \textbf{0.838} & \textbf{0.917} & \textbf{0.755} & 0.737 & 0.854 & 0.917 & 0.782 \\
         & $(1, 3)$ & 0.737 & 0.847 & 0.917 & 0.760 & 0.738 & 0.847 & 0.917 & 0.775 \\
         & $(1, 5)$ & 0.733 & 0.837 & 0.916 & 0.761 & 0.735 & 0.847 & 0.917 & 0.774 \\
         & $(2, 1)$ & 0.735 & 0.845 & 0.917 & 0.756 & 0.736 & 0.852 & 0.916 & 0.782 \\
         & $(5, 1)$ & 0.736 & 0.847 & 0.916 & 0.759 & 0.735 & 0.853 & 0.917 & 0.780 \\
         & $(5, 2)$ & 0.734 & 0.854 & 0.916 & 0.767 & 0.733 & 0.844 & 0.916 & 0.777 \\
         & $(5, 5)$ & 0.734 & 0.845 & 0.916 & 0.757 & 0.735 & 0.846 & 0.917 & 0.774 \\
         & $(5, 10)$ & 0.735 & 0.845 & 0.917 & 0.758 & \textbf{0.736} & \textbf{0.844} & \textbf{0.916} & \textbf{0.769} \\
        
        \bottomrule
    \end{tabular}
    }
\end{table}

\end{document}